\documentclass[journal,twoside,web,twocolumn]{ieeecolor}
\usepackage{generic}
\usepackage{cite}
\usepackage{amsmath,amssymb,amsfonts}
\usepackage{algorithmic}
\usepackage{graphicx}
\usepackage{algorithm,algorithmic}
\usepackage{hyperref}
\hypersetup{hidelinks=true}
\usepackage{textcomp}
\usepackage{ulem}
\usepackage{cancel}
\usepackage{mathtools}
\usepackage{tikz}
\newtheorem{remark}{Remark}

\def\BibTeX{{\rm B\kern-.05em{\sc i\kern-.025em b}\kern-.08em
    T\kern-.1667em\lower.7ex\hbox{E}\kern-.125emX}}
\newtheorem{theorem}{Theorem}
\newtheorem{lemma}{Lemma}
\newtheorem{corollary}{Corollary}
\newtheorem{definition}{Definition}
\newtheorem{proposition}{Proposition}

\renewcommand*{\emph}[1]{\textit{#1}}

\begin{document}
\title{A Minimal Control Family of Dynamical Systems for Universal Approximation}
\author{Yifei Duan, Yongqiang Cai
\thanks{The authors are with the School of Mathematical Sciences, Laboratory of Mathematics and Complex Systems, MOE, Beijing Normal University, 100875 Beijing, China (Email:
caiyq.math@bnu.edu.cn.)}}

\maketitle

\begin{abstract}
    The universal approximation property (UAP) holds a fundamental position in deep learning, as it provides a theoretical foundation for the expressive power of neural networks. It is widely recognized that a composition of linear and nonlinear functions, such as the rectified linear unit (ReLU) activation function, can approximate continuous functions on compact domains. In this paper, we extend this efficacy to a scenario containing dynamical systems with controls. We prove that the control family $\mathcal{F}_1$ containing all affine maps and the nonlinear ReLU map is sufficient for generating flow maps that can approximate orientation-preserving (OP) diffeomorphisms on any compact domain. Since $\mathcal{F}_1$ contains only one nonlinear function and the UAP does not hold if we remove the nonlinear function, we call $\mathcal{F}_1$ a minimal control family for the UAP. On this basis, several mild sufficient conditions, such as affine invariance, are established for the control family and discussed. Our results reveal an underlying connection between the approximation power of neural networks and control systems and could provide theoretical guidance for examining the approximation power of flow-based models.
	
\end{abstract}


\section{Introduction}

The universal approximation property (UAP) of neural networks plays a crucial role in the field of deep learning. This property implies that given sufficient parameters, a neural network can approximate any continuous function with arbitrary precision \cite{Cybenkot1989Approximation, Hornik1989Multilayer, Leshno1993Multilayer, Lu2017Expressive}. Although polynomials and one-hidden-layer networks also exhibit the UAP, the preference for deep network structures in practical deep learning situations arises from the superior performance and expressive capabilities of deep networks. Some studies have shown that deep neural networks are more expressive than shallow neural networks are \cite{telgarsky2016benefits, daniely2017depth}, but this is far from a complete understanding of the properties of composition functions.

Mathematically, the flow map of a dynamical system provides a natural framework for studying function composition. A natural idea is whether we can draw inspiration from the flow maps of dynamical systems to study the properties of composite functions such as deep neural networks. This idea not only holds theoretical significance but also practical relevance. In the field of engineering, it was initially discovered that residual networks (ResNets) with skip connections could achieve superior performance to that of other networks \cite{he2016deep,he2016identity}. Later, mathematicians recognized that this type of structure is closely connected to dynamical systems \cite{Weinan2017A, haber2017stable}: a ResNet can be viewed as the forward Euler discretization of a dynamical system. This connection has given rise to many novel network architectures \cite{Luo2022Rethinking}, such as 
PolyNet \cite{Zhang2017PolyNet}, RevNet \cite{Gomez2017Reversible} and LM-ResNet \cite{lu2018beyond}. In particular, continuous-time residual networks, known as neural ordinary differential equations (ODEs), have garnered widespread attention \cite{Chen2018Neural}. Compared to ResNet, neural ODE is more conducive to analysis from the perspective of dynamical systems.

From the perspective of approximation theory, both residual networks and neural ODEs are universal approximators \cite{Lin2018ResNet, Teshima2020Couplingbased, Zhang2020Approximation, Li2022Deep, Tabuada2022Universal, Ruiz-Balet2023Neural}. In particular, control theory can be employed as a tool for analyzing the approximation power of dynamical system flow maps \cite{Cuchiero2020Deep, Li2022Deep, Tabuada2022Universal, Cheng2023Interpolation}. For a dynamical system $\dot{x}(t) = f(x(t))$ with an initial value $x(0)=x_0 \in \mathbb{R}^d$, by taking the field function $f$ as a one-hidden-layer feedforward network with a nonlinear activation function $\sigma$, such as the hyperbolic tangent function, one obtains the following neural ODE:
\begin{align}
    \label{eq:NODE_tanh}
    \dot{x}(t) = f_{\theta(t)}(x(t)) = S(t) \sigma(W(t)x+b(t)),
\end{align}
where $\theta(t) := (S(t),W(t),b(t))\in \mathbb{R}^{d \times N}\times \mathbb{R}^{N\times d} \times \mathbb{R}^N  \eqqcolon \Theta$ are time-varying control parameters and $N$ is the number of hidden neurons. The results of \cite{Tabuada2022Universal} showed that the flow map of Eq.~\eqref{eq:NODE_tanh} with $N=d$ can uniformly approximate monotonic functions if the utilized activation function satisfies a quadratic differential equation. When conducting an approximation under the $L^p$ norm with $p\in[1,\infty)$, taking the function $\sigma$ as a Lipschitz nonlinear function is sufficient for approximating continuous functions if the dimensionality $d\ge2$ \cite{Li2022Deep}. In addition, if $\sigma$ is a so-called ``well function'', the parameters $S$ and $W$ can be restricted to diagonal matrices with entries whose absolute values are less than one \cite{Li2022Deep}. However, in these approach, the control family $\mathcal{F}:=\{ f_{\theta} ~|~ \theta \in \Theta\}$
is complex enough in that its span already has rich approximation power.
In this work, we explore the possibility of imposing weaker conditions on control families such that the compositions of their flow maps are universal approximators.
On the basis of a novel route implemented via orientation-preserving (OP) diffeomorphisms and novel techniques involving splitting methods, we obtain some weaker conditions but stronger results (which will be observed later).

The flow maps of a smooth dynamical system are OP diffeomorphisms. The inverse states that any OP diffeomorphisms of $\mathbb{R}^d$ can be represented or approximated by composing a finite number of flow maps. Let $\{f_1,...,f_m\}$ be a set that satisfies some conditions (\emph{e.g.,} let it be a Lie bracket-generating family according to Definition \ref{def:Lie_family}) of vector fields on $\mathbb{R}^d$.
The results of \cite{Agrachev2010Dynamics, Caponigro2011Orientation} show that there exist $m$ functions $u_1(t,x),...,u_m(t,x)$ called controls, which are
piecewise constant in $t$ and smooth in $x\in\mathbb{R}^d$, such that the flow $\Phi$ of the following affine control system~(\ref{eq:affine_control_system_varying}) at time 1 (\emph{i.e.,} the map from $x(0)$ to $x(1)$) can uniformly approximate any OP diffeomorphism $\Psi \in \text{Diff}_0(\mathbb{R}^d)$ on any compact domain $\Omega$:
\begin{align}
    \label{eq:affine_control_system_varying}
    & \dot x(t) = \sum_{i=1}^{m} u_i(t,x) f_i(x).
\end{align}
Although the set of all OP diffeomorphisms of $\mathbb{R}^d$, $\text{Diff}_0(\mathbb{R}^d)$, is a small subset of $C(\mathbb{R}^d,\mathbb{R}^d)$ (the set of continuous functions from $\mathbb{R}^d$ to $\mathbb{R}^d$), it is dense in $C(\Omega,\mathbb{R}^d)$ under the $L^p$ norm (but not dense under the $L^\infty$ norm or uniform norm) for any compact domain $\Omega \subset \mathbb{R}^d$ provided that the dimensionality $d$ is larger than one \cite{Brenier2003Approximation}.
Following this derivation, to approximate continuous functions, one only needs to approximate the flows of system~\eqref{eq:affine_control_system_varying} or approximate the corresponding OP diffeomorphisms. In this paper, we answer the following question: \emph{how simple can a control family $\mathcal{F}$ be while allowing the compositions of its flows to arbitrarily approximate any OP diffeomorphisms?}

We show that the expected control family can be very simple, such as $\mathcal{F} = \mathcal{F}_1$ for $d \ge 2$, which contains affine maps and one nonlinear function:
\begin{align*}
    \mathcal{F}_1 = \{x \mapsto Ax+b: A\in \mathbb{R}^{d\times d}, b \in \mathbb{R}^d\} \cup \{ \text{ReLU}(\cdot)\},
\end{align*}
where ReLU is an elementwise function that maps $x=(x_1,...,x_d)$ to $ \text{ReLU}(x)=(\max(x_1,0),...,\max(x_d,0))$. As an implication of this simple control family, there exists a family $\{f_1,...,f_m\}$ such that the flow of the affine control system~\eqref{eq:affine_control_system_varying} can approximate OP diffeomorphisms even if all the controls $u_1(t,x),...,u_m(t,x)$ are independent of $x$, \emph{i.e.}, $u_i(t,x) = u_i(t)$.
This approximation property complements traditional control theory, where the focus is on examining the controllability of a system \cite{Boscain2019Introduction}.

By exploring the UAP of flow map compositions and leveraging concepts from control theory, the results of this study can enhance our understanding of the expressive power and approximation capabilities of function compositions. Overall, the integration of control theory, dynamical systems, and deep learning provides a promising direction for advancing our knowledge and unlocking the full potential of deep compositions in various domains.

Our contributions are as follows.
\begin{enumerate}
    \item We show that a very simple control family that contains only affine maps and one nonlinear function is sufficient for generating flow maps that can uniformly approximate the diffeomorphisms of $\mathbb{R}^d$ on any compact domain (Theorem \ref{th:UAP_ReLU_family}). This control family is the minimal family for the UAP in the sense that it contains only one nonlinear function, whereas the affine maps alone do not satisfy the UAP.
    \item We provide some other sufficient conditions imposed on the control family such that the hypothesis space satisfies the UAP (Theorem \ref{th:UAP_span_F_to_HF} and Theorem \ref{th:UAP_affine}). Our results reveal an underlying connection between the approximation power of neural networks and the flow maps of control systems.
    \item Our theorems imply many mild UAP-related conditions that are easy to verify. This can provide theoretical guidance for examining the approximation power of flow-based models (e.g., \cite{Chen2018Neural,Li2022Deep,Tabuada2022Universal}).
\end{enumerate}

This paper is organized as follows. In Section \ref{sec:main}, we present the notations and our main results concerning the UAP. The proofs of the theorems are provided in Section~\ref{sec:proof}.
We discuss the possible extensions and applications of our results in Section~\ref{sec:discussion} and provide a summary in Section~\ref{sec:summary}.

\section{Notations and main results}
\label{sec:main}

\subsection{Universal approximation property}

In this work, all functions are assumed to be continuous and restricted to compact domains. In addition, we define the UAP as follows \cite{Cai2024Vocabulary}.
\begin{definition}[Universal approximation property, UAP]
    For any compact domain $\Omega \subset \mathbb{R}^d$, target function space $\mathcal{T}$ and hypothesis space $\mathcal{H}$, we say that
    1) $\mathcal{H}$ has a $C(\Omega)$-UAP for $\mathcal{T}$ if for any $g \in \mathcal{T}$ and $\varepsilon>0$, there is a function $h \in \mathcal{H}$ such that $\|g-h\|_{C(\Omega)} \le \varepsilon$, \emph{i.e.}
        \begin{align}
            \|g(x)-h(x)\|_\infty \le \varepsilon, \quad \forall x \in \Omega.
        \end{align}
        
     2) $\mathcal{H}$ has an $L^p(\Omega)$-UAP for $\mathcal{T}$ with $p \in [1,+\infty)$ if for any $g \in \mathcal{T}$ and $\varepsilon>0$, there is a function $h \in \mathcal{H}$ such that
        \begin{align} 
            \|g-h\|_{L^p(\Omega)} = \Big(\int_\Omega \|g(x)-h(x)\|_\infty^p dx\Big)^{1/p} \le \varepsilon.
        \end{align}
\end{definition}
In the absence of ambiguity, we can omit $\Omega$ to say the ``$C$-UAP'' and ``$L^p$-UAP''. If the target space $\mathcal{T} = C(\Omega)$ (or $L^p(\Omega)$) and $\mathcal{H} \subset \mathcal{T}$, then saying that $\mathcal{H}$ has a $C$-UAP (or an $L^p$-UAP) for $\mathcal{T}$ is equivalent to saying that $\mathcal{H}$ is dense in $\mathcal{T}$.

Note that here we use the term ``$C$-UAP'' instead of ``$L^\infty$-UAP'' for two reasons. 
(1) $C$ and $L^p$ represent both norms and the target function spaces; the norms $L^\infty$ and $C$ have subtle differences, and the space $L^\infty(\Omega, \mathbb{R}^d)$ is significantly different from $C(\Omega, \mathbb{R}^d)$. 
(2) The $L^\infty$-UAP is stronger than the $L^p$-UAP, $L^\infty$ is not used so that we can avoid having to specifically emphasize that $p$ does not include $\infty$ when referring to $L^p$.

Notably, the $C$-UAP implies the $L^p$-UAP since we are considering functions on compact domains. The well-known Stone-Weierstrass theorem \cite{stone1948generalized} indicates that polynomials satisfy the $C$-UAP for continuous functions. In addition, neural networks, such as the following one-hidden-layer network $h(x)$, also follow the $C$-UAP for continuous functions \cite{Leshno1993Multilayer}:
\begin{align}
    \label{eq:ReLU-NN}
    h(x; s,w,b) = 
    \sum\limits_{i=1}^N s_i \sigma(w_i \cdot x+b_i),
\end{align}
where $N$ is the number of neurons, $\sigma(\cdot)$ is a nonpolynomial activation function, and $s_i, w_i \in \mathbb{R}^{d}, b_i \in \mathbb{R}$ are parameters. Note that neural networks achieve the UAP via function composition, which is different from the linear combination method used for polynomials. According to the definition of continuous functions, we have the following proposition \cite[Lemma 3.11]{Duan2022Vanilla}, which indicates that the UAP can be inherited via function compositions.

\begin{proposition}
    \label{th:composition_approximation}
    Let the map $T = F_n \circ ... \circ F_1$ be a composition of $n$ continuous functions $F_i$ defined on open domains $D_i$, and let $\mathcal{F}$ be a continuous function class that can uniformly approximate each $F_i$ on any compact domain $\mathcal{K}_i \subset D_i$ (under the uniform norm). Then, for any compact domain $\mathcal{K} \subset D_1$ and $\varepsilon >0$, there are $n$ functions $\tilde F_1, ..., \tilde F_n$ in $\mathcal{F}$ such that
    $\|T(x) - \tilde F_n \circ ... \circ \tilde F_1 (x)\|_\infty \le \varepsilon$ for all $x \in \mathcal{K}$.
\end{proposition}

\subsection{Control families and diffeomorphisms}
Here, we recall the flow maps of dynamical systems and their related diffeomorphisms. We consider the following ODE in dimension $d$:
\begin{align}\label{eq:ODE_general}
    \left\{
    \begin{aligned}
    &\dot{x}(t) = v(x(t),t), \quad t\in(0,\tau), \tau>0,\\
    &x(0)=x_0 \in \mathbb{R}^d.
    \end{aligned}
    \right.
\end{align}
where the vector-valued field function $v(x,t)$ is Lipschitz continuous in $x$ and piecewise constant in $t$, which ensures that the corresponding system has a unique solution $x(t)$ for any initial value $x_0$. The map from $x_0$ to $x(\tau)$ is called a \emph{flow map}, and we denote it by $\phi_v^\tau(x_0)$. Later, we consider the case in which all $v(\cdot,t), t\ge 0,$ belong to a prescribed set $\mathcal{F}$. When $v(x,t)=f(x)$ is independent of $t$, it is easy to check that the associated flow map has the following group property for composition.
\begin{proposition}
    Let $f(\cdot)$ be Lipschitz continuous, and let $\phi_f^t$ be the flow map of $\dot x(t)=f(x(t))$; then, $\phi_f^{0}$ is the identity map, and $\phi_f^{t_1}\circ \phi_f^{t_2} = \phi_f^{t_1+t_2}$ for any $t_1, t_2 \ge 0$.
\end{proposition}
In addition, when $v(x,t)$ is piecewise constant in $t$, the flow map $\phi_v^\tau$ can be represented as the following composition of $\phi_{f_i}^{\tau_i}, \tau_i \ge 0$:
\begin{align}
    \label{eq:flow_map_composition}
    \phi_{v(\cdot,t)}^\tau(x) 
    =
    \phi_{f_n(\cdot)}^{\tau_n} \circ \cdots \circ
    \phi_{f_2(\cdot)}^{\tau_2} \circ
    \phi_{f_1(\cdot)}^{\tau_1} (x),
\end{align}
where $\tau = \sum_{i=1}^n \tau_i$ and $v(x,t) = f_i(x)$ if $t$ is within the interval $\big(\sum_{j=1}^{i-1} \tau_j, \sum_{j=1}^i \tau_j \big)$.

A differentiable map $\Phi: \mathbb{R}^d \to \mathbb{R}^d$ is called a diffeomorphism if it is a bijection and its inverse $\Phi^{-1}$ is differentiable as well. In addition, a diffeomorphism $\Phi$ of $\mathbb{R}^d$ is called OP if its Jacobian (the determinant of the Jacobian matrix) $\det(\nabla \Phi(x))$ is positive for all $x \in \mathbb{R}^d$. As mentioned earlier, $\text{Diff}_0(\mathbb{R}^d)$ has the $L^p$-UAP for $C(\Omega,\mathbb{R}^d)$ when $d\ge 2$. However, $\text{Diff}_0(\mathbb{R}^d)$ has no $C$-UAP for $C(\Omega,\mathbb{R}^d)$ because some continuous functions, such as $g(x_1,x_2)=(x_1^2+x_2^2,0), \Omega=[-2,2]^2$, cannot be arbitrarily approximated by diffeomorphisms under the uniform norm (see \cite{Cai2023Achieve} for example).

The classic results show that the flow map ${\phi_v^\tau}$ is an OP diffeomorphism if the vector function ${v}$ is twice continuously differentiable with respect to $x$ \cite[Ch.4]{Arnold1992Ordinary}. In contrast, OP diffeomorphisms can be uniformly approximated by the flow maps of ODEs \eqref{eq:ODE_general} with smooth field functions $v(x,t)$ \cite{Agrachev2010Dynamics}. To state this result more clearly, we introduce the definition of a \emph{control family} and its deduced hypothesis space below.
\begin{definition}[Control family]
    We call a set $\mathcal{F} = \{f_\theta : \mathbb{R}^d \to \mathbb{R}^d | \theta \in \Theta\}$ of continuous functions parameterized by $\theta \in \Theta$ a control family. In addition, we say that
    \begin{enumerate}
        \item $\mathcal{F}$ is symmetric if $f \in \mathcal{F}$ implies that $-f \in \mathcal{F}$,
        \item $\mathcal{F}$ is affine invariant if $f \in \mathcal{F}$ implies that $ D f(A\cdot +b) \in \mathcal{F}$ for any matrices $D,A \in \mathbb{R}^{d\times d}$ and $b \in \mathbb{R}^{d}$,
        \item $\mathcal{F}$ is diagonal affine invariant if $f \in \mathcal{F}$ implies that $ D f(A\cdot +b) \in \mathcal{F}$ for any diagonal matrices $D,A \in \mathbb{R}^{d\times d}$ and $b \in \mathbb{R}^d$,
        \item and $\mathcal{F}$ is restricted affine invariant \cite{Li2022Deep} if $f \in \mathcal{F}$ implies that $ D f(A\cdot +b) \in \mathcal{F}$ for any diagonal matrices $D, A \in \mathbb{R}^{d\times d}$ and $b \in \mathbb{R}^d$, where the entries of $D$ are $\pm 1$ or 0, and the entries of $A$ have absolute values that are smaller than or equal to 1.

\end{enumerate}
\end{definition}

In this paper, we focus on the case in which all $ f \in \mathcal{F}$ are Lipschitz continuous and denote the deduced hypothesis space $\mathcal{H}(\mathcal{F})$ as the set of all compositions of flow maps $\phi^\tau_f, f \in \mathcal{F}$:
\begin{align*}
    \mathcal{H}(\mathcal{F})
    =
    \{
    \phi_{f_n}^{\tau_n} \circ \cdots \circ
    \phi_{f_2}^{\tau_2} \circ
    \phi_{f_1}^{\tau_1} ~|~
    f_i \in \mathcal{F}, \tau_i \ge 0, n \in \mathbb{N}
    \}.
\end{align*}
According to the relation shown in Eq.~\eqref{eq:flow_map_composition}, $\mathcal{H}(\mathcal{F})$ can be rewritten as
\begin{align*}
    \mathcal{H}(\mathcal{F})
    =
    \{
    & \phi^\tau_{f_{\theta(t)}} : x \mapsto z(\tau) 
    ~|~ \notag \\
    &\dot z(t) = f_{\theta(t)}(z(t)), f_{\theta(t)} \in \mathcal{F}, z(0)=x , \tau \ge 0\},
\end{align*}
where $\theta(t)$ is a piecewise constant function. We return to the theorem of \cite{Agrachev2010Dynamics}, which shows that an OP diffeomorphism can be approximated by $\mathcal{H}(\mathcal{P})$, where $\mathcal{P}$ is a set of polynomials. Since their proof only uses the density of polynomials, their theorem can be directly extended to general cases as follows.
\begin{proposition}
    \label{th:UAP_F_to_HF}
    If a control family $\mathcal{F}$ is dense in $C(\mathcal{K},\mathbb{R}^d)$ for every compact domain $\mathcal{K} \subset \mathbb{R}^d$, then the hypothesis space $\mathcal{H}(\mathcal{F})$
has the $C$-UAP for $\text{Diff}_0(\mathbb{R}^d)$.
\end{proposition}

\subsection{Main results}

Now, we present some sufficient conditions for a Lipschitz continuous control family $\mathcal{F}$ such that $\mathcal{H}(\mathcal{F})$ follows the $C$-UAP for the OP diffeomorphisms of $\mathbb{R}^d$ ($\text{Diff}_0(\mathbb{R}^d)$) and hence obeys the $L^p$-UAP for continuous functions $C(\mathbb{R}^d,\mathbb{R}^d)$ when $d\ge2$.

First, we extend Proposition~\ref{th:UAP_F_to_HF} to a milder condition that ensures the UAP.
\begin{theorem}
    \label{th:UAP_span_F_to_HF}
    Let $\Omega \subset \mathbb{R}^d$ be a compact domain, and let $\mathcal{F}$ be a symmetric Lipschitz control family. If the span of $\mathcal{F}$, \emph{i.e.},
    \begin{align}
        \text{span}(\mathcal{F}) = \Big\{ \sum_{i=1}^m a_i f_i ~|~ a_i \in \mathbb{R}, f_i \in \mathcal{F}, m \in \mathbb{N} \Big \},
    \end{align}
    is dense in $C(\mathcal{K},\mathbb{R}^d)$ for every compact domain $\mathcal{K} \subset \mathbb{R}^d$, then $\mathcal{H}(\mathcal{F})$ has the $C(\Omega)$-UAP for $\text{Diff}_0(\mathbb{R}^d)$.
\end{theorem}

The proof of Theorem~\ref{th:UAP_span_F_to_HF} is given in Section~\ref{sec:proof_UAP_span_F_to_HF}.

Then, we apply Theorem~\ref{th:UAP_span_F_to_HF} to affine invariant control families. Consider a control family $\mathcal{F}_{\text{aff}}(f)$ that is affine invariant and contains a nonlinear Lipschitz continuous function $f$:
\begin{align*}
    \mathcal{F}_{\text{aff}}(f)
    =
    \{ x \mapsto D f(Ax+b) ~|~ D, A\in \mathbb{R}^{d\times d}, b \in \mathbb{R}^d\}.
\end{align*}
This $\mathcal{F}_{\text{aff}}(f)$ is the smallest affine invariant family containing $f \in C(\mathbb{R}^d,\mathbb{R}^d)$. Similarly, we denote $\mathcal{F}_{\text{diag}}(f)$ and $\mathcal{F}_{\text{rest}}(f)$ as the smallest diagonal and a restricted affine invariant family containing $f$, respectively. In addition, we introduce the following definition of nonlinearity.
\begin{definition}
    For a (vector-valued) continuous function $f:=(f_1,...,f_{d'}) \in C(\mathbb{R}^d,\mathbb{R}^{d'})$ from $\mathbb{R}^d$ to $\mathbb{R}^{d'}$ with components $f_i \in C(\mathbb{R}^d,\mathbb{R}), i= 1,...,d'$, we say that
    \begin{enumerate}
        \item $f$ is nonlinear if at least one component $f_i$ is nonlinear,
        \item $f_i$ is coordinate nonlinear if for any coordinate vector $e_j \in \mathbb{R}^d, j=1,...,d$, there is a vector $b_{ij} \in \mathbb{R}^d$ such that $f_i(\cdot e_j + b_{ij}) \in C(\mathbb{R},\mathbb{R})$ is nonlinear,
        \item and $f$ is fully coordinate nonlinear if all components $f_i$ are coordinate nonlinear.
    \end{enumerate}
\end{definition}

The result provided in \cite{Li2022Deep} shows that for a specific function \(f \), which is referred to as a \emph{well function}, the hypothesis space \(\mathcal{H}(\mathcal{F}_{\text{rest}}(f))\) generated by the restricted affine invariant family \(\mathcal{F}_{\text{rest}}(f)\) follows the \(L^p\)-UAP for continuous functions when \(d \geq 2\). In addition, \cite{Tabuada2022Universal} showed that for some $f$ satisfying a quadratic differential equation, the hypothesis space $\mathcal{H}(\mathcal{F}_{\text{aff}}(f))$ has the $C$-UAP for monotonical functions. In this paper, we extend these results by showing that allowing \(f \) to be (fully coordinate) nonlinear is sufficient for achieving the UAP. Specifically, the following theorem generalizes the results of \cite{Li2022Deep} and \cite{Tabuada2022Universal}.
\begin{theorem}
    \label{th:UAP_affine}
    Let $f \in C(\mathbb{R}^d,\mathbb{R}^d)$ be a nonlinear Lipschitz continuous function; then, $\mathcal{H}(\mathcal{F}_{\text{aff}}(f))$ has the $C$-UAP for $\mathrm{Diff}_0(\mathbb{R}^d)$. In addition, if $f$ is Lipschitz continuous and fully coordinate nonlinear, then both $\mathcal{H}(\mathcal{F}_{\text{diag}}(f))$ and $\mathcal{H}(\mathcal{F}_{\text{rest}}(f))$ has the $C$-UAP for $\mathrm{Diff}_0(\mathbb{R}^d)$.
\end{theorem}
This theorem indicates that the expressive power differences among $\mathcal{F}_{\text{aff}}(f)$, $\mathcal{F}_{\text{diag}}(f)$ and $\mathcal{F}_{\text{rest}}(f)$ are subtle. The first part of Theorem~\ref{th:UAP_affine} is implied from Theorem~\ref{th:UAP_span_F_to_HF} and the classic UAP for one-hidden-layer neural networks. The proof of this part is given in Section~\ref{sec:proof_UAP_affine}. However, the second part involves the construction of many techniques based on Theorem~\ref{th:UAP_ReLU_family}, which are described later. The proof of this part is given in Section~\ref{sec:proof_UAP_diag_affine}.

Finally, we provide a surprisingly simple control family that has the UAP. Let us start with the linear 
\footnote{
	We do not strictly differentiate between the terms ``linear'' and ``affine'', as they are often used interchangeably in machine learning. For instance, a linear layer in neural networks is actually an affine map. Additionally, a nonlinear map essentially refers to a non-affine map.
} 
control family $\mathcal{F}_0$, which contains only affine maps, \emph{i.e.,}  
\begin{align}
    \mathcal{F}_0 = \{ x \mapsto Ax+b ~|~ A\in \mathbb{R}^{d\times d}, b \in \mathbb{R}^d\}.
\end{align}
Notably, $\mathcal{F}_0$ is affine invariant, and $\mathcal{H}(\mathcal{F}_0)$ only contains linear OP diffeomorphisms, which implies that $\mathcal{H}(\mathcal{F}_0)$ does not have the UAP for nonlinear functions. However, adding one nonlinear function $\text{ReLU}(\cdot)$ to $\mathcal{F}_0$ significantly improves the resulting approximation power. Consider the following control family $\mathcal{F}_1$ and its symmetric version $\mathcal{F}_2$:
\begin{align}
    \mathcal{F}_1 = \mathcal{F}_0 \cup \{ \text{ReLU}(\cdot)\}, \quad
    \mathcal{F}_2 = \mathcal{F}_0 \cup \{ \pm \text{ReLU}(\cdot) \}.
\end{align}
The UAP of these families is stated in Theorem~\ref{th:UAP_ReLU_family}.
\begin{theorem}
    \label{th:UAP_ReLU_family}
    $\mathcal{H}(\mathcal{F}_2)$ has the $C$-UAP for $\mathrm{Diff}_0(\mathbb{R}^d)$. $\mathcal{H}(\mathcal{F}_1)$ has the $C$-UAP for $\mathrm{Diff}_0(\mathbb{R}^d)$ if and only if $d\ge 2$.
\end{theorem}
Theorem~\ref{th:UAP_ReLU_family} is nontrivial, as $\text{span}(\mathcal{F}_1)=\text{span}(\mathcal{F}_2)$ is rather simple and far from the condition stated in Theorem~\ref{th:UAP_span_F_to_HF} (as $\text{span}(\mathcal{F}_1)=\text{span}(\mathcal{F}_2)$ is a finite-dimensional space and not dense in $C(\mathcal{K},\mathbb{R}^d)$ for any compact domain $\mathcal{K}\subset \mathbb{R}^d$). We provide a proof of this theorem in Section~\ref{sec:proof_UAP_ReLU_family}; this proof is based on Proposition~\ref{th:UAP_F_to_HF} and motivated by recent work on the relationships between dynamical systems and feedforward neural networks \cite{Duan2022Vanilla}. To approximate an OP diffeomorphism $\Phi \in \mathrm{Diff}_0(\mathbb{R}^d)$, our construction process, which we call a \emph{flow-discretization-flow} approach, contains three steps. 
(i) $\Phi$ is approximated by a flow map $\phi^\tau_h$, where $h(x; s,w,b)$ is the leaky rectified linear unit (leaky-ReLU) network in Eq.~\eqref{eq:ReLU-NN} whose parameters $s(t), w(t), b(t)$ are piecewise constant in $t$. 
(ii) $\phi^\tau_h$ is approximated by discretizing it with a proper splitting method. 
(iii) It is verified that the discretized mappings are flow maps in $\mathcal{H}(\mathcal{F}_1)$.

In summary, three sufficient conditions for the UAP are outlined.
In particular, the condition in Theorem \ref{th:UAP_ReLU_family} states that a control family needs only a single nonlinear function along with affine maps. As an implication, if one wants to verify that a control family $\mathcal{F}$ possesses the UAP, one only needs to check whether $\text{span}(\mathcal{F})$ can approximate linear maps and the ReLU function. For instance, one can verify that Theorem~\ref{th:UAP_ReLU_family} holds if we replace the $\text{ReLU}$ function in $\mathcal{F}_1$ and $\mathcal{F}_2$ with a broad class of functions, such as the softplus function $\tilde \sigma_0(x) = \ln(1+\text{e}^x)$ and the map $\hat \sigma_0 : (x_1,...,x_d) \mapsto \text{ReLU}(x_d) e_d$, which keeps only one ReLU component nonzero.

\section{Theorem proofs}
\label{sec:proof}

\subsection{Candidate flow maps}
In this paper, we focus on the flow map $\phi_v^t$, where $t \ge 0$ and $v(\cdot,t)$ is Lipschitz continuous on $\mathbb{R}^d$. Although the flow map $\phi_v^t$ is defined on the whole space $\mathbb{R}^d$, we mainly consider its restriction to certain compact domains. Table~\ref{tab:flowmaps} lists some examples of flow maps. The flow maps for $v = \text{ReLU}=\sigma_0$ are related to the leaky-ReLU functions denoted by $\sigma_\alpha, \alpha \ge 0$:
\begin{align}
    \label{eq:leaky-ReLU}
    &\sigma_\alpha(x) = \sigma_\alpha(x_1,...,x_d) = 
    \big(
        \sigma_\alpha(x_1),...,\sigma_\alpha(x_d)
    \big),\\
    &\sigma_\alpha(x_i) = 
    \begin{cases}
        x_i, & x_i \ge 0, \\
        \alpha x_i , & x_i <0.
    \end{cases} 
\end{align}
In addition, let $A\in \mathbb{R}^{d\times d}$ be nonsingular and $b \in \mathbb{R}^{d}$; then, it is easy to verify the following relations:
\begin{align}
    &\phi^t_{A^{-1} \circ v \circ A} = A^{-1} \circ \phi^t_v \circ A, 
    \\
    &\phi^t_{v(\cdot +b)} = (\cdot - b) \circ \phi^t_{v(\cdot)} \circ (\cdot + b).
\end{align}
\begin{table}[htbp]
    \caption{Examples of flow maps.}
    \centering
\begin{tabular}{|c|c|}
    \hline 
    Dynamical system & Flow map \\
    \hline
    $\dot x(t) = Ax(t),\quad A \in \mathbb{R}^{d\times d}$ & $\phi_{A \cdot}^t: x \mapsto e^{At} x$ \\
    \hline
    $\dot x(t) = b, \quad b \in \mathbb{R}^d$ & $\phi_b^t: x \mapsto x + bt$ \\
    \hline
    $\dot x(t) = \sigma_0(x(t))$ & $\phi_{\sigma_0(\cdot)}^t : x \mapsto e^t \sigma_{e^{-t}}(x)$ \\
    \hline
    $\dot x(t) = -\sigma_0(x(t))$ & $\phi_{-\sigma_0(\cdot)}^t : x \mapsto e^{-t} \sigma_{e^t}(x)$ \\
    \hline
\end{tabular}
\label{tab:flowmaps}
\end{table}

Next, we characterize the flow maps in the hypothesis spaces $\mathcal{H}(\mathcal{F}_0)$, $\mathcal{H}(\mathcal{F}_1)$ and $\mathcal{H}(\mathcal{F}_2)$.
$\mathcal{H}(\mathcal{F}_0)$ is simple, whereas $\mathcal{H}(\mathcal{F}_1)$ and $\mathcal{H}(\mathcal{F}_2)$ are complex. In the following subsections, we consider three typical types of maps that are useful for our later constructions.

\subsubsection{Linear flow maps}

Here, we characterize a flow map $\phi_ f^\tau$ with a linear vector field $f \in \mathcal{F}_0$ and map it in $\mathcal{H}(\mathcal{F}_0)$. Table~\ref{tab:flowmaps} shows that if $e^A=:W \in \mathbb{R}^{d\times d}$ is nonsingular and its matrix logarithm $\ln(W)=A$ is real, then the linear map $W: x \mapsto Wx$ is the linear flow map $\phi^1_{A \cdot}$ in $\mathcal{H}(\mathcal{F}_0)$. Remark~\ref{th:real_log_matrix} below gives some examples of such a matrix $W$.
\begin{remark}
    \label{th:real_log_matrix}
    The following versions of matrix $W$ have real logarithms:
    \begin{enumerate}
        \item[(1)] $W = P \tilde W P^{-1}$, where $\tilde W$ has a real logarithm and $P$ is nonsingular,
        \item[(2)] $W = \Lambda:= \text{diag}(\lambda_1,...,\lambda_d)$ is diagonal with $\lambda_i >0$,
        \item[(3)] $W = \Lambda_{\pm} := \text{diag}(\lambda_1,...,\lambda_d)$ is diagonal with $\lambda_i=\pm1$ and has a positive determinant $\det(\Lambda_{\pm})>0$,
        \item[(4)] and $W = U := I + \lambda E_{ij}$ with $\lambda \in \mathbb{R}$ and $i\neq j$.
    \end{enumerate}
\end{remark}
Points (1), (2) and (4) are obvious since $\ln(P \tilde W P^{-1}) = P \ln(\tilde W) P^{-1}$, $\ln(\Lambda) = \text{diag}(\ln(\lambda_1),...,\ln(\lambda_d))$, and $\ln(U) = U - I$. Point (3) is implied by the simplest case in which $d=2$:
\begin{align}
    \ln \Lambda_2
        :=
    \ln
        \left(
    \begin{matrix} 
        -1 & 0\\ 
        0 & -1
    \end{matrix}
    \right)
    =
    \left(
    \begin{matrix} 
        0 & \pi\\ 
        -\pi & 0
    \end{matrix}
    \right)
    =:A_2.
\end{align}
For the case in which $d>2$, there must be an even number of $\lambda_i=-1$ (otherwise $\det(\Lambda_{\pm})<0$), and hence, there exists a permutation matrix $P$ such that $P \Lambda_{\pm} P^{-1}$ is a block diagonal matrix with $\Lambda_2$ blocks and ones. Then, $\ln(\Lambda_{\pm})$ is real, as is $\ln(P \Lambda_{\pm} P^{-1})$.
\begin{lemma}
    \label{th:represent_HF0}
    Let $W \in \mathbb{R}^{d\times d}$ have a positive determinant and $b\in\mathbb{R}^d$; then, the map $\Phi: x \mapsto Wx +b$ belongs to $\mathcal{H}(\mathcal{F}_0)$.
\end{lemma}
\begin{proof}
    Since $W$ is nonsingular, linear algebra shows that there exists a sequence of elementary matrices $U_k,k=1,...,n$ of the form $U_k = I + \lambda_k E_{i_k,j_k}, \lambda_k \in \mathbb{R}, i_k \neq j_k$, such that $M = U_n \cdots U_2 U_1 W$ is diagonal. In addition, $\det(M)=\det(W) >0$, and $M$ can be represented as a multiplication of two diagonal matrices $M = \Lambda \Lambda_{\pm}$, where $\Lambda$ has positive diagonal entries and $\Lambda_{\pm}$ has $\pm 1$ diagonal entries. According to Remark~\ref{th:real_log_matrix}, the map $\Psi: x \mapsto Wx$ belongs to $\mathcal{H}(\mathcal{F}_0)$. In addition, $\Phi = \phi^1_b \circ \Psi$ also belongs to $\mathcal{H}(\mathcal{F}_0)$.
\end{proof}

As a consequence, we can characterize $\mathcal{H}(\mathcal{F}_0)$ as
\begin{align*}
\mathcal{H}(\mathcal{F}_0) = \big\{x \mapsto Wx+b ~|~ W \in \mathbb{R}^{d\times d},\det(W)>0, b\in\mathbb{R}^d \big\}.
\end{align*}

\subsubsection{ReLU flow maps}
To characterize the maps contained in $\mathcal{H}(\mathcal{F}_1)$ and $\mathcal{H}(\mathcal{F}_2)$, we generalize the leaky-ReLU functions described in Eq.~\eqref{eq:leaky-ReLU} to the vector parameter case. Letting $\vec{\boldsymbol \alpha} = (\alpha_1,...,\alpha_d) \in \mathbb{R}^d$, we define a map $\sigma_{\boldsymbol{\alpha}} : \mathbb{R}^d \to \mathbb{R}^d$ with parameters $\vec{\boldsymbol \alpha}$ as
\begin{align}
    \sigma_{\vec{\boldsymbol \alpha}}(x) = 
    \sigma_{\vec{\boldsymbol \alpha}}(x_1,...,x_d) = 
    \big(\sigma_{\alpha_1}(x_1),..., \sigma_{\alpha_d}(x_d) \big).
\end{align}
If all $\alpha_i$ are equal to positive numbers $\alpha$, then $\sigma_{\vec{\boldsymbol \alpha}} = \sigma_{\alpha}$. Table~\ref{tab:flowmaps} shows that the leaky-ReLU function $\sigma_{\alpha} = \phi^t_{-I \cdot} \circ \phi^t_{\sigma_0}$ (here, $I \cdot$ refers to the identity map), with $t=-\ln(\alpha)$, belonging to $\mathcal{H}(\mathcal{F}_1)$ if $\alpha \in (0,1)$ and $\sigma_{\alpha} \in \mathcal{H}(\mathcal{F}_2)$ if $\alpha > 0$. For the vector parameter case, we have the following lemma.
\begin{lemma}
    \label{th:ReLU_flow}
    Let $d\ge 2$ and $\vec{\boldsymbol \alpha} = (\alpha_1,\cdots,\alpha_d) \in (0,\infty)^d$; then, for any compact domain $\mathcal{K} \subset \mathbb{R}^d$, there is a map $\Phi \in \mathcal{H}(\mathcal{F}_1)$ such that $\sigma_{\vec{\boldsymbol \alpha}} = \Phi$ on $\mathcal{K}$. 
\end{lemma}
\begin{proof}
    Let $\vec{\boldsymbol \alpha}_1 = (\alpha_1,1,...,1)$, ..., $\vec{\boldsymbol \alpha}_i$, ..., $\vec{\boldsymbol \alpha}_d = (1,...,1,\alpha_d)$ be $d$ vectors that have ones as their elements at all positions except for location $i$. It is obvious that $\sigma_{\vec{\boldsymbol \alpha}}$ is a composition of $\sigma_{\vec{\boldsymbol \alpha}_i}$:
        $\sigma_{\vec{\boldsymbol \alpha}} = \sigma_{\vec{\boldsymbol \alpha}_d}
        \circ \cdots \circ 
        \sigma_{\vec{\boldsymbol \alpha}_1}$.
    Therefore, we only need to prove that each $\sigma_{\vec{\boldsymbol \alpha}_i}$ belongs to $\mathcal{H}(\mathcal{F}_1)$. Since each $\sigma_{\vec{\boldsymbol \alpha}_i}$ has the same structure over a permutation, we only prove the case of $\sigma_{\vec{\boldsymbol \alpha}_1}$. Without loss of generality, we consider the case in which $d=2$. If $\alpha_1=1$, then the map is an identity, which is a flow map layer. Next, we consider the cases in which $ \alpha_1 \in (0,1)$ and $\alpha_1 >1$.

    (1) $ \alpha_1 \in (0,1)$. The norm of any $x\in\mathcal{K}$ is bounded by a constant $C$. We construct the map $\Phi$ as the following composition of three maps:
    \begin{align}
        \Phi: 
        \left(
        \begin{matrix} 
            x_1 \\ 
            x_2
        \end{matrix}
        \right)
		\mapsto
        \left(
        \begin{matrix} 
            x_1 \\ 
            x_2 + C
        \end{matrix}
        \right) 
        & 
		\mapsto
        \sigma_{\alpha_1}
        \left(
        \begin{matrix} 
            x_1 \\ 
            x_2 + C
        \end{matrix}
        \right)
        =
        \left(
        \begin{matrix} 
            \sigma_{\alpha_1}(x_1) \\ 
            x_2 + C
        \end{matrix}
        \right) \nonumber\\
        & \mapsto
        \left(
        \begin{matrix} 
            \sigma_{\alpha_1}(x_1) \\ 
            x_2
        \end{matrix}
        \right).
    \end{align}
    The first and third maps are shift maps. The second map is exactly $\sigma_{\alpha_1}$. Here, we use the fact that $\sigma_{\alpha_1}(x_2+C) = x_2+C$ for $x\in\mathcal{K}$. It is easy to verify that these three maps are flow maps; therefore, $\Phi \in \mathcal{H}(\mathcal{F}_1)$ and $\sigma_{\alpha_1}=\Phi$ on $\mathcal{K}$.

    (2) $\alpha_1 >1$. Because $\sigma_{\alpha_1}(x) = -\alpha_1 \sigma_{\beta_1}(-x)$, where $\beta_1=1/\alpha_1 \in (0,1)$, the function $\Phi$ is constructed as the following composition of four maps:
    \begin{align}
        \label{eq:ReLU_flow}
        \Phi: 
        \left( \begin{matrix}  x_1 \\  x_2 \end{matrix} \right) 
		& 
		\mapsto
        \left( \begin{matrix}  -x_1 \\  -x_2 \end{matrix} \right) \mapsto
        \left( \begin{matrix} 
            \sigma_{\beta_1}(-x_1) \\ 
            -x_2 
        \end{matrix} \right) 
		\mapsto
        \left( \begin{matrix} 
            -\sigma_{\beta_1}(-x_1) \\ 
            x_2 
        \end{matrix} \right) \nonumber \\ 
        & \mapsto
        \left(
        \begin{matrix} 
            \sigma_{\alpha_1}(x_1) \\ 
            x_2
        \end{matrix}
        \right).
    \end{align}
    The first, third, and fourth functions are linear maps corresponding to $\text{diag}(-1,-1)$ and $\text{diag}(\alpha_1,1)$. The second function is exactly $\sigma_{\vec{\boldsymbol\beta}_1}$ with $\vec{\boldsymbol\beta}_1 = (\beta_1,1)$. According to case (1) and Lemma~\ref{th:represent_HF0}, the constructed $\Phi$ is what we need.
\end{proof}

\subsubsection{Two-piecewise linear maps}
Here, we consider a special two-piecewise linear function $T^{(j)} \in C(\mathbb{R}^d,\mathbb{R}^d)$, which is useful for proving our main theorems. The map 
$T^{(j)}: x=(x^{(1)},...,x^{(d)}) \mapsto y=(y^{(1)},...,y^{(d)}), j=1,...,d$ 
\begin{align}
    \label{eq:two_piecewise_linear}
    \left\{
        \begin{aligned} 
        & y^{(i)} = x^{(i)} ,  i \neq j,\\
        & y^{(j)} = x^{(j)} + a  \sigma_{\alpha}(w_1 x^{(1)}+ \cdots +w_{d} x^{(d)} + \beta),
        \end{aligned}
    \right.
\end{align}
where $w=(w_1,...,w_d) \in \mathbb{R}^d, a, \beta \in \mathbb{R}$ and $\alpha >0$. The parameters of the function are chosen such that both $1+aw_j$ and $1+\alpha a w_j$ are positive, while $T^{(j)}$ is a one-to-one map. Lemma~\ref{th:T_as_flowmap} shows that the function $T^{(j)}$ belongs to $\mathcal{H}(\mathcal{F}_1)$.
\begin{lemma}
    \label{th:T_as_flowmap}
    Let $d\ge 2, \alpha>0$ and $\max(1,\alpha)|a w_d| < 1$; then, for any compact domain $\mathcal{K} \subset \mathbb{R}^d$, there is a map $\Phi \in \mathcal{H}(\mathcal{F}_1)$ such that $\Phi = T^{(j)}$ on $\mathcal{K}$.
\end{lemma}
\begin{proof}
    We only prove the case in which $j=d$ and simply denote $T^{(d)}$ as $T=T^{(d)}$. The case in which $j \neq d$ can be proved  similarly by relabeling the utilized index. The case in which $\|w\|=0$ is trivial; hence, we assume that $\|w\|\neq 0$. The bias $\beta$ can be absorbed by introducing two shift maps; hence, we only need to consider the case in which $\beta=0$. In fact, let $b = \beta w / \|w\|^2$; then, we have that $T = \tilde T(x+b)-b$, where $\tilde T$ has the same form as that of $T$ and zero bias.

(1) The case in which $\sum_{i=1}^{d-1} w^2_i = 0 $.
It is easy to check that $x^{(d)} + a \sigma_\alpha(w_d x^{(d)})$ 
is proportional to a generalized leaky-ReLU function of $x^{(d)}$;
then, $T$ can be represented as $T= \text{diag}(1,...,1,\lambda) \circ \sigma_{\vec{\boldsymbol \gamma}}$, $\vec{\boldsymbol \gamma} = (1,...,1,\gamma)$, where $\lambda$ and $\gamma$ are
\begin{align}
    \begin{cases}
        \lambda = (1 + a w_d) , \gamma = (1+\alpha a w_d)/\lambda, & \text{if} w_d > 0, \\
        \lambda = (1 + \alpha a w_d), 
        \gamma =  (1+ a w_d)/\lambda,& \text{if } w_d <0.
    \end{cases}
\end{align}
Then, $T$, which is restricted to $\mathcal{K}$, equals a map $\Phi$ in $\mathcal{H}(\mathcal{F}_1)$ according to Lemma~\ref{th:ReLU_flow}.

(2) The case in which $\sum_{i=1}^{d-1} w^2_i \neq 0 $. Without loss of generality, we assume that $w_1 \neq 0$. In addition, using $\sigma_\alpha(\cdot) = - \alpha \sigma_{1/\alpha}(-\cdot)$ again, we can further assume that $w_1>0$.
The map $T$ can be represented by the following composition:
\begin{align}
    T \equiv F_5 \circ \cdots \circ F_0,
\end{align}
where each mapping step is as follows:
\begin{align*} 
    \left(
        \begin{matrix} x^{(1)} \\ x^{(2:d-1)} \\ x^{(d)}
        \end{matrix}\right)
    &\underrightarrow{F_0}
    \left(
        \begin{matrix} \nu \\ x^{(2:d-1)}\\ x^{(d)}
        \end{matrix}\right)  
    \underrightarrow{F_1}
    \left(
        \begin{matrix} \sigma_\alpha{(\nu)} \\ x^{(2:d-1)}\\ x^{(d)}
        \end{matrix}\right)\\
    &\underrightarrow{F_2}
    \left(
        \begin{matrix} \sigma_\alpha{(\nu)} \\ x^{(2:d-1)}\\ x^{(d)}+ a \sigma_\alpha{(\nu)} 
        \end{matrix}\right)
    \underrightarrow{F_3}
    \left(
        \begin{matrix} \nu \\ x^{(2:d-1)}\\ x^{(d)} + a \sigma_\alpha{(\nu)} 
        \end{matrix}\right)\\
    &\underrightarrow{F_4}
    \left(
        \begin{matrix} \nu+w_d a \sigma_\alpha{(\nu)} \\ x^{(2:d-1)}\\ x^{(d)} + a \sigma_\alpha{(\nu)} 
        \end{matrix}\right)
    \underrightarrow{F_5}
    \left(
        \begin{matrix} x^{(1)} \\ x^{(2:d-1)}\\ x^{(d)} + a \sigma_\alpha{(\nu)} 
        \end{matrix}\right).
\end{align*}
Here, $\nu=w_1 x^{(1)}+\cdots+w_dx^{(d)}$, and $x^{(2:d-1)}$ represents the elements $x^{(2)},...,x^{(d-1)}$.

Note that $F_0,F_2,F_5$ are linear maps in $\mathcal{H}(\mathcal{F}_0)$:
\begin{align}
    &F_0(x) = \left(
        \begin{matrix} 
            w_1 &w_{2:d}\\ 
            0 & I_{d-1}
        \end{matrix}
        \right) x, \quad
    F_2(x) = \left(
        \begin{matrix} 
            I_{d-1} &0\\ 
            (a,0_{2:d-1}) & 1
        \end{matrix}
        \right) x, \nonumber\\
    & F_5(x) = \left(
        \begin{matrix} 
            1/w_1 &-w_{2:d}/w_1\\ 
            0 & I_{d-1}
        \end{matrix}
        \right) x,
\end{align}
where $I_{d-1}$ is the identity matrix and $(a,0_{2:d-1})=(a,0,...,0)$ has $d-2$ zeros. $F_1$ is a leaky-ReLU map $\sigma_{\vec{\boldsymbol \alpha}}$ with $\vec{\boldsymbol \alpha}=(\alpha,1,...,1)$, $F_3$ is a leaky-ReLU map $\sigma_{\vec{\boldsymbol \beta}}$ with $\vec{\boldsymbol \beta}=(1/\alpha,1,...,1)$, and $F_4$ is a flow map of $F_4=\text{diag}(1+a w_d,1,...,1) \circ \sigma_{\vec{\boldsymbol \gamma}}$ with $\vec{\boldsymbol \gamma}=((1+\alpha a w_d a)/(1+a w_d),1,...,1)$. According to Lemma~\ref{th:represent_HF0} and Lemma~\ref{th:ReLU_flow}, we clarify that each component $F_i,i=0,\cdots,5,$ restricted to any compact domain is a flow map in $\mathcal{H}(\mathcal{F}_1)$. As a result, $T$, which is restricted to $\mathcal{K}$, is a flow map in $\mathcal{H}(\mathcal{F}_1)$.
The proof is now complete.
\end{proof}

\subsection{Error estimation and splitting methods}

Here, we demonstrate that the flow maps of two systems are close to each other if their vector field functions are sufficiently close.
\begin{lemma}
    \label{th:ODE_error_estimation}
    Consider two ODE systems
    \begin{align}\label{eq:ODE_system}
        \dot{x}(t) = f_i(x(t)), \quad t\in(0,\tau), \quad i=1,2,
    \end{align}
    where the $f_i(x)$ are continuous in $x\in\mathbb{R}^d$. In addition, we assume that $f_1(x)$ is $L$-Lipschitz continuous, \emph{i.e.,} $\|f_1(x)-f_1(x')\| \le L \|x-x'\|$ for any $x,x' \in \mathbb{R}^d$. Then, for any compact domain $\Omega$ and $\varepsilon>0$, there exists a $\delta \in (0,1]$ such that the flow maps $\phi_{f_1}^\tau(x) ,\phi_{f_2}^\tau(x)$ satisfy
    \begin{align}
        \|\phi_{f_1}^\tau(x)- \phi_{f_2}^\tau(x)\| \le \varepsilon,
        \quad \forall x\in\Omega,
    \end{align}
    provided that $\|f_1(x)-f_2(x)\|< \delta$ for all $x \in \Omega_\tau$, where $\Omega_\tau$ is a compact domain defined as $\Omega_\tau = \{ x + (V+1) \tau e^{L\tau} x' ~|~ x \in \Omega, \|x'\| \le 1, V = \max_{x\in \Omega}\{\|f_1(x)\|\} \}$.
\end{lemma}

It is well known that an ODE system can be approximated via many numerical methods. In particular, we use the splitting approach \cite{McLachlan2002Splitting}. Given an ODE $\dot{x}(t)=v(x(t),t)$ with $x(0)=x_0$, where $v(x,t)$ is split into the summation of several functions,
\begin{align}
    v(x,t) = \sum_{i=1}^m v_i(x,t), m \in \mathbb{Z}^+.
\end{align}
For a given time step $\Delta t$, we define $x_k$ 
as the following iterative equation:
\begin{align}
    x_{k+1} = T_k (x_k) = T_k^{(m)} \circ \dots \circ T_k^{(2)} \circ T_k^{(1)} (x_k),
\end{align}
where the map $T_k^{(i)}: x \to y$ (this map is related to the map $T^{(j)}$ from \eqref{eq:two_piecewise_linear}, which is discussed later in the second part of the proof of Theorem~\ref{th:UAP_ReLU_family}) is given by the forward Euler discretization shown below:
\begin{align}\label{eq:splitting}
    y = T_k^{(i)}(x) = x + \Delta t v_i(x,t_k),
    \quad
    t_k = k \Delta t.
\end{align}
The following lemma demonstrates the convergence analysis of the splitting method.

\begin{lemma}[Convergence of the splitting method]
    \label{th:split_method}
    Let all $v_i(x,t), (x,t)\in \mathbb{R}^d \times [0,\tau], i=1,2,...,m,$ be piecewise constant in $t$ and Lipschitz continuous in $x$ with a shared Lipschitz constant $L>0$. Then, for any $\varepsilon>0$ and a compact domain $\Omega$, there exists a positive integer $n$ and a $\Delta t = \tau/n$ such that $\|x(\tau) - x_n\| \le \varepsilon$ for all $x_0 \in \Omega$.
\end{lemma}

The main message of Lemma~\ref{th:ODE_error_estimation} and Lemma~\ref{th:split_method} is that one can approximate a flow map with a complex field by compositing flow maps with some simpler fields. The proof is standard, and we leave it for the Appendix.

\subsection{Proof of Theorem~\ref{th:UAP_span_F_to_HF}}
\label{sec:proof_UAP_span_F_to_HF}

We begin with a lemma for composing flow maps. For any natural number $n$, we denote $f^{\circ n} = f\circ f \circ \cdots \circ f$ as the function that composites $f$ for $n$ times.
\begin{lemma}
    \label{th:flow_sum_f}
    Let $\mathcal{F}=\{f_1,\cdots f_m\}$ be a finite control family, where $f_i:\mathbb{R}^d\to \mathbb{R}^d$ denotes Lipschitz functions; then, for any compact domain $\Omega \subset \mathbb{R}^d$, any function
    $f = \sum_{i=1}^m a_i f_i$ with $a_i \in \mathbb{R}^+$, and $\tau \ge 0$, the flow map $\phi_f^\tau$ is in the closure of $\mathcal{H}(\mathcal{F})$; 
    \emph{i.e.}, $\phi_f^\tau \in {\text{cl} (\mathcal{H}(\mathcal{F}))}$ under the topology of $C(\Omega)$.
    
\end{lemma}
\vspace{12pt}
\begin{proof}
    The lemma is implied by the well-known product formula \cite{chorin1978product}; that is, for any $\varepsilon>0$, there exists a sufficiently large integer $n$ such that
    \begin{align}
        \label{eq:product_approximation}
        \big \|\phi_{\sum_{i=1}^m a_i f_i}^\tau(x)-
        \big(\phi_{f_m}^{{a_m \tau}/{n}} \circ \cdots \phi_{f_1}^{{a_1 \tau}/{n}} \big)^{\circ n}(x)
        \big\| \le \varepsilon, \notag \\ \forall x\in \Omega.
    \end{align} 
    The technique of this proof is similar to that used for Lemma~\ref{th:split_method}, but here, the discretization map
    $T_k^{(i)}$ in Eq.~\eqref{eq:splitting} is independent of $k$ and is replaced by its continuous version $\phi_{a_i f_i}^{\Delta t}, \Delta t = \tau/n$. Given that $\phi_{f_i}^{{ a_i\tau}/{n}} = \phi_{a_i f_i}^{{\tau}/{n}}$ as $a_i \ge 0$, we can absorb $a_i$ with $f_i$; hence, we can assume that $a_i=1$ later. Note that the difference between $T_k^{(i)}$ and $\phi_{f_i}^{\Delta t}$ is on the order of $O(\Delta t^2)$, and the error accumulation when calculating $x_k$ by replacing all $T_k^{(i)}$ with $\phi_{f_i}^{\Delta t}$ is bounded by $O(\text{e}^{m L k\Delta t} \Delta t)$. In fact, we have
    \begin{align*}
        \|T_k^{(i)}(x) - \phi_{f_i}^{\Delta t}(x)\| 
        &\le
        \int_0^{\Delta t} \| f_i(x) - f_i(\phi_{f_i}^{s}(x)) \| ds \notag \\
        &\le L M \text{e}^{L \Delta t} (\Delta t)^2, \quad \forall x \in \Omega_\tau,
    \end{align*}
    and 
    \begin{align*}
        \big\|x_n - \big(\phi_{f_m}^{{\tau}/{n}} \circ \cdots \phi_{f_1}^{{\tau}/{n}} \big)^{\circ n}(x_0) \big\|
        \le 
        M \text{e}^{m L \tau} \text{e}^{L \Delta t} \Delta t, \notag \\
        \forall x_0 \in \Omega,
    \end{align*}
    where $\Omega_\tau$ is a compact set such that the flows starting from $x_0 \in \Omega$ are contained in $\Omega_\tau$, and $M$ is a constant that is dependent on $\Omega_\tau$ and $f_i$. Hence, with Lemma~\ref{th:split_method}, when $n$ is sufficiently large such that $\Delta t=\frac{\tau}{n}$ is small enough, the inequality Eq.~\eqref{eq:product_approximation} can be obtained.
\end{proof}

\begin{proof}[Proof of Theorem~\ref{th:UAP_span_F_to_HF}]
    According to Proposition~\ref{th:UAP_F_to_HF}, $\mathcal{H}(\text{span}(\mathcal{F}))$ has the $C(\Omega)$-UAP for $\text{Diff}_0(\mathbb{R}^d)$. Therefore, we only need to show that ${\mathcal{H}(\mathrm{span}(\mathcal{F}))} \subset {\text{cl}(\mathcal{H}(\mathcal{F}))}$. Since the UAP can be inherited via function composition, as shown in Proposition~\ref{th:composition_approximation}, it is sufficient to show that $\phi_g^\tau \in {\text{cl}(\mathcal{H}(\mathcal{F}))}$ for any $g \in \text{span}(\mathcal{F})$. This is directly implied by Lemma~\ref{th:flow_sum_f}, as $\mathcal{F}$ is symmetric, where $g$ can be represented as a linear combination of some $f_i \in \mathcal{F}$ with positive coefficients.
\end{proof}

\subsection{Proof of Theorem~\ref{th:UAP_affine} (the first part)}
\label{sec:proof_UAP_affine}

Let us recall the classic UAP possessed by neural networks: one-hidden-layer networks obey the $C$-UAP for continuous functions. Specifically, if an activation function $\sigma: \mathbb{R}\to\mathbb{R}$ is nonlinear and Lipschitz continuous (which implies that $\sigma$ is nonpolynomial), then for any compact domain $\mathcal{K} \subset \mathbb{R}^d$, $g\in C(\mathcal{K}, \mathbb{R}^d)$ and $\varepsilon>0$, there exists a one-hidden-layer network $h(x)$ with $N$ neurons such that
\begin{align}
    \label{eq:NN_one_hidden_layer}
    \|h(x)-g(x)\|
    =
    \Big \|\sum\limits_{i=1}^N s_i \sigma(w_i \cdot x+b_i) - g(x) \Big\|
    \le \varepsilon, \notag \\ 
    \forall x \in \mathcal{K},
\end{align}
where $s_i, w_i \in \mathbb{R}^d, b_i \in \mathbb{R}$ are parameters. Note that the nonlinearity of $f$ allows us to design the employed activation function $\sigma$ accordingly.
\begin{lemma}
    \label{th:nonlinear_function}
    Let $f_s \in C(\mathbb{R}^d,\mathbb{R})$ be a scalar nonlinear function; then, there are two vectors $a,b \in \mathbb{R}^d$ such that $\sigma(\cdot) = f_s(\cdot a + b) \in C(\mathbb{R},\mathbb{R})$ is nonlinear.
\end{lemma}
\begin{proof}
    Employing a proof by contradiction, we prove that a function $g$ must be linear if all $g(\cdot a + b), a,b \in \mathbb{R}^d,$ are linear. In fact, $g$ is determined by its value at the coordinate vectors $e_1,...,e_d$ and the origin. We define $h(\cdot) = g(\cdot) - g(0)$; then, the partial linearity of $g$ and $h$ shows that
    \begin{align}
        h( (a-b) t + b ) = (h(a) - h(b)) t + h(b), \notag \\
        \forall a,b \in \mathbb{R}^d, t\in\mathbb{R}.
    \end{align}
    In particular, $h(x_i e_i) = x_i h(e_i)$ for any $x_i \in \mathbb{R}$. Choosing $t=1/2$, we obtain the following:
    \begin{align*}
        h( x_1 e_1 + x_2 e_2 )
        &=
        h( (2x_1 e_1 - 2x_2 e_2)t + 2x_2 e_2) \\
        &= (h(2x_1 e_1) - h(2x_2 e_2)) t + h(2x_2 e_2) \\ 
        & = x_1 h( e_1) + x_2 h(e_2 ).
    \end{align*}
    Similarly, we can show that 
    $h(x_1 e_1 +...+ x_d e_d)
    = x_1 h( e_1) +...+ x_d h(e_d)$, which indicates that $h$ and hence $g$ are linear functions.
\end{proof}

Note that the above result does not hold if the vector $a$ is restricted to being a coordinate vector. Setting $f(x_1,...,x_d)=x_1 x_2 \cdots x_d$ offers such a counterexample.

On the basis of Lemma~\ref{th:nonlinear_function}, we can provide a proof for the first part of Theorem~\ref{th:UAP_affine}.

\begin{proof}[Proof of Theorem~\ref{th:UAP_affine} (the first part)]
    We prove that $\mathcal{H}(\mathcal{F}_{\text{aff}}(f))$ has the $C(\Omega)$-UAP for $\mathrm{Diff}_0(\mathbb{R}^d)$ if $f \in C(\mathbb{R}^d,\mathbb{R}^d)$ is Lipschitz continuous and nonlinear. Employing Theorem~\ref{th:UAP_span_F_to_HF}, we only need to show that $\text{span}(\mathcal{F}_{\text{aff}}(f))$ is dense in $C(\mathcal{K},\mathbb{R}^d)$ for every compact domain $\mathcal{K} \subset \mathbb{R}^d$.

Since $f=(f_1,...,f_d)$ is nonlinear, at least one of its components is nonlinear. Without loss of generality, we assume that $f_1(x)=f_1(x_1,x_2,...,x_d)$ is nonlinear. In addition, the affine invariance property and Lemma~\ref{th:nonlinear_function} allow us to further assume that the activation function $\sigma: \mathbb{R}\to\mathbb{R}$ defined by $\sigma(\cdot)=f_1(\cdot e_1)=e_1^T f(\cdot,0,...,0)$ is nonlinear. Then, any $g\in C(\mathcal{K}, \mathbb{R}^d)$ can be approximated by the one-hidden-layer network $h(x)$ described in Eq.~\eqref{eq:NN_one_hidden_layer}, which becomes
    \begin{align}
        \label{eq:h_one_hidden_sigma}
        h(x)
        = 
        \sum\limits_{i=1}^N s_i \sigma(w_i \cdot x+b_i)
        =
        \sum\limits_{i=1}^N s_i e_1^T f\big( e_1 w_i^T x+b_i e_1 \big).
    \end{align}
    Since $s_i e_1^T, e_1 w_i^T \in \mathbb{R}^{d\times d}$ and $b_i e_1 \in \mathbb{R}^d$, we have that $h \in \text{span}(\mathcal{F}_{\text{aff}}(f))$. Hence, the proof is finished based on the UAP of neural networks.
\end{proof}
Note that the idea of the above proof cannot be directly employed to prove the second part of Theorem~\ref{th:UAP_affine}. In other words, the continuous function $g\in C(\mathcal{K}, \mathbb{R}^d)$ or $h$ in Eq.~\eqref{eq:h_one_hidden_sigma} cannot belong to the closure $\text{cl(span}(\mathcal{F}_{\text{diag}}(f)))$. For example, if $f=\text{ReLU}, d=2, \mathcal{K}=[-1,1]^2$, then $\text{cl}({\text{span}(\mathcal{F}_{\text{diag}}(\text{ReLU}))})$ only contains functions ${v}: \mathbb{R}^2 \to \mathbb{R}^2$ with the form {$v(x_1,x_2) = (v_1(x_1),v_2(x_2)))$}. This setting fails to approximate various functions, such as $g(x_1,x_2)=(x_1 x_2, 0)$, in $C(\mathcal{K}, \mathbb{R}^2)$. However, when $d\ge 2$, the ReLU function is nonlinear but not fully coordinate nonlinear (\emph{e.g.,} the first component of $\text{ReLU}(x_1,x_2)$ is $\text{ReLU}(x_1)$, which is constant and linear for the second coordinate $x_2$); hence, we need more techniques to address the second part of Theorem~\ref{th:UAP_affine}.

\subsection{Proof of Theorem~\ref{th:UAP_ReLU_family}}
\label{sec:proof_UAP_ReLU_family}

We consider the cases in which $d=1$ and $d\ge 2$ separately.

\begin{lemma}
    \label{th:HF1and2_d=1}
    Let $d=1$ and $\Omega \subset \mathbb{R}$ be a compact interval. Then, $\mathcal{H}(\mathcal{F}_2)$ has the $C(\Omega)$-UAP for $\mathrm{Diff}_0(\mathbb{R})$, and there exists a function $\Phi \in \mathrm{Diff}_0(\mathbb{R})$ that cannot be approximated by the functions contained in $\mathcal{H}(\mathcal{F}_1)$.
\end{lemma}
\begin{proof}
    Note that the functions in $\mathrm{Diff}_0(\mathbb{R})$ are strictly monotonically increasing. However, it is easy to check that $\mathcal{H}(\mathcal{F}_1)$ contains only convex functions since $\phi^\tau_{A\cdot+b}$ and $\phi^\tau_{\sigma_0}$ are monotonically increasing convex functions and the composition of increasing convex functions is also convex. Hence, $\mathcal{H}(\mathcal{F}_1)$ does not have the $C(\Omega)$-UAP for $\mathrm{Diff}_0(\mathbb{R})$; for example, the concave function $\Phi(x)= x-e^{-x}$ in $\mathrm{Diff}_0(\mathbb{R})$ cannot be approximated by the functions contained in $\mathcal{H}(\mathcal{F}_1)$.

For the hypothesis space $\mathcal{H}(\mathcal{F}_2)$, we adopt the results of \cite{Duan2022Vanilla}, which shows that a monotonic function can be approximated by leaky-ReLU networks. That is, for any $\Psi \in \mathrm{Diff}_0(\mathbb{R})$ and $\varepsilon>0$, there exists a leaky-ReLU network $g(x)$ with a width of one such that
    $
    \|\Psi-g\|_{C(\Omega)} \le \varepsilon.
    $
    Notably, $g(x)$ is also a monotonically increasing function, and the number of negative weight parameters included in $g$ must be even; otherwise, $g$ is a monotonically decreasing function. By utilizing the relation $\sigma_\alpha(-x)=-\alpha\sigma_{1/\alpha}(x)$, we can offset the negative weights in $g$ and deduce that $g\in \mathcal{H}(\mathcal{F}_2)$.
    Hence, $\Psi$ can be approximated by functions in $\mathcal{H}(\mathcal{F}_2)$, and the $C$-UAP of $\mathcal{H}(\mathcal{F}_2)$ is proved.
\end{proof}

Before presenting the proof for the case in which $d\ge 2$, we introduce our \emph{flow-discretization-flow} approach (motivated by \cite{li2023minimum}), which contains three steps. First, motivated by \cite{agrachev2009controllability, Agrachev2010Dynamics, Caponigro2011Orientation}, we approximate the OP diffeomorphism via the flow map of neural ODEs. We then discretize the neural ODEs with the splitting method, which approximates the flow maps by composing a sequence of functions (motivated by \cite{Duan2022Vanilla, li2023minimum}). Finally, we clarify that all the employed composition functions are flow maps in $\mathcal{H}(\mathcal{F}_1)$.

\begin{proof}[Proof of Theorem~\ref{th:UAP_ReLU_family}]
    According to Lemma~\ref{th:HF1and2_d=1}, we only need to prove that $\mathcal{H}(\mathcal{F}_1)$ has the $C(\Omega)$-UAP for $\mathrm{Diff}_0(\mathbb{R}^d)$ when $d\ge 2$. For any OP diffeomorphism $\Psi$ of $\mathbb{R}^d$, a compact domain $\Omega$ and $\varepsilon>0$, we construct a function $\phi \in \mathcal{H}(\mathcal{F}_1)$ according to the {flow-discretization-flow} approach such that $\|\Psi - \phi\|_{C(\Omega)} \le \varepsilon$.

(1) Approximating $\Psi$ via a flow map of the associated neural ODEs. According to Proposition~\ref{th:UAP_F_to_HF} and the UAP of neural networks, there exists a one-hidden-neural network $v(x,t)$ with $N$ hidden neurons and a leaky-ReLU activation function $\sigma_\alpha, \alpha \in (0,1)$:
    \begin{align}
        v(x,t) = 
        \sum\limits_{i=1}^N s_i(t)\sigma_\alpha(w_i(t) \cdot x+b_i(t))
    \end{align}
    where $s_i, w_i \in \mathbb{R}^d, b_i \in \mathbb{R}$ are piecewise constant functions in $t$ such that their flow map $\phi_v^\tau, \tau=1,$ approximates $\Psi$ uniformly,
    \emph{i.e.}, 
    \begin{align}
        \|\Psi(x_0)-\phi_v^\tau(x_0)\| \le \varepsilon/2, 
        \quad \forall x_0 \in \Omega.
    \end{align}

    (2) Discretizing $\phi_v^\tau$ via the splitting method. The vector-valued function $v$ can be split into $Nd$ summed functions $v(x,t) = \sum_{i=1}^N \sum_{j=1}^d v_{ij}(x,t) e_j$, where $e_j$ is the $j$-th coordinate vector and the $v_{ij}(x,t) = s_{ij}(t) \sigma_\alpha(w_{i}(t) \cdot x+b_i(t))$ are scalar Lipschitz functions. According to Lemma~\ref{th:split_method}, there is a time step $\Delta t = \tau/n$, where $n$ is sufficiently large such that
    \begin{align}
        \|\Phi(x_0) - \phi_v^\tau(x_0) \| \le \varepsilon/2,
        \quad \forall x_0 \in \Omega,
    \end{align}
    where the function $\Phi$ is defined via iteration:
    \begin{align}
        \label{eq:composition_split}
        x_n &= \Phi(x_0) = T_n \circ \cdots \circ T_1 (x_0) \notag \\
        &=
        T_n^{(N,d)} \circ \dots \circ T_n^{(1,2)} \circ T_n^{(1,1)} \notag \\ 
        &~~~~\circ
        \cdots \circ
        T_1^{(N,d)} \circ \cdots \circ T_1^{(1,2)} \circ T_1^{(1,1)} 
        (x_0).
    \end{align}
    Here, the $T^{(i,j)}_k$ in $\Phi$ are smooth and Lipschitz continuous on $\mathbb{R}^d$ and are represented by $T_k^{(i,j)}: x=(x^{(1)},...,x^{(d)}) \mapsto y=(y^{(1)},...,y^{(d)})$ with
    \begin{align}\label{eq:map_T}
        \left\{
        \begin{aligned} 
        &y^{(l)} =  x^{(l)} , \quad l \neq j, \\
        &y^{(j)} =  x^{(j)} + \Delta t v_{ij}(x,t_k),
        \end{aligned}
        \right.
    \end{align}
    where $v_{ij}(x,t_k)=s_{ij}(t_k) \sigma_\alpha(w_{i}(t_k) \cdot x+b_i(t_k))$.

    (3) Representing each $T_k^{(i,j)}$ as a flow map in $\mathcal{H}(\mathcal{F}_1)$. This is achieved according to Lemma~\ref{th:T_as_flowmap}. In fact, we define
    \begin{align}
        M = \sup\{ |s_{ij}(t) w_{ij}(t)| ~|~ i,j=1,...,d, t \in [0,\tau] \},
    \end{align}
    and let $\Delta t < 1/M$ or $n >M$; then, the condition imposed by Lemma~\ref{th:T_as_flowmap} is satisfied. As a consequence, $\Phi \in \mathcal{H}(\mathcal{F}_1)$. Combining (1) and (2), we have that $\|\Psi - \Phi\|_{C(\Omega)} \le \varepsilon$, which finishes the proof.
\end{proof}

\subsection{Proof of Theorem~\ref{th:UAP_affine} (the second part)}
\label{sec:proof_UAP_diag_affine}

\begin{proof}[Proof of Theorem~\ref{th:UAP_affine} (the second part)]
    We consider $\mathcal{F}_{\text{diag}}(f)$ and $\mathcal{F}_{\text{rest}}(f)$ separately.

    (1) $\mathcal{H}(\mathcal{F}_{\text{diag}}(f))$ has the $C(\Omega)$-UAP for $\mathrm{Diff}_0(\mathbb{R}^d)$.
    
    Based on Theorem~\ref{th:UAP_ReLU_family}, we need to prove that
    \begin{align}
        \mathcal{H}(\mathcal{F}_1) \subset \text{cl}(\mathcal{H}(\mathcal{F}_{\text{diag}}(f)));
    \end{align}
    according to Lemma \ref{th:ODE_error_estimation} and Lemma \ref{th:flow_sum_f}, we only need to show that
    \begin{align}
        \mathcal{F}_1 \subset {\text{cl(span}(\mathcal{F}_{\text{diag}}(f)))}.
    \end{align}
    Notably, all $g=(g_1,\cdots,g_d)$ in $\mathcal{F}_1$ can be represented in the following summation form:
    \begin{align}
        g(x) = 
        \sum\limits_{i=1}^d \sum\limits_{j=1}^d g_i(x_j e_j) e_i
        =
        \sum\limits_{i=1}^d \sum\limits_{j=1}^d g_{ij}(x_j) e_i
    \end{align}
    where the $g_{ij}(x_j) = g_i( x_j e_j)$ are univariate functions (either ReLU or affine functions).

    Similar to the proof of the first part, here, the full nonlinearity of the current situation allows us to define $d^2$ nonlinear activation functions $\sigma_{ij}(\cdot) = e_i^T f(\cdot e_j + \hat b_{ij}), \hat b_{ij} \in \mathbb{R}^d,$ and construct neural networks that can approximate $g_{ij}(x_j)$. In fact, $g_{ij}(x_j)$ can be approximated by a neural network $h_i(x_j)$ with $N_{ij}$ neurons, the $\sigma_{ij}$ activation function, and parameters $s_{ijk},w_{ijk},b_{ijk} \in \mathbb{R}$; this neural network is described as follows:
    \begin{align}
        h_i(x_j)
        &= 
        \sum\limits_{k=1}^{N_{ij}}  s_{ijk} \sigma_{ij}(w_{ijk} x_j+b_{ijk}).
    \end{align}
    Hence, $g(x)$ can be approximated by $h(x)$:
    \begin{align}
        \label{eq:h_in_diagonal_affine}
        h(x) :&= 
        \sum_{i=1}^d \sum_{j=1}^d h_i(x_j) e_i \notag \\
        & =
        \sum_{i=1}^d \sum_{j=1}^d \sum_{k=1}^{N_{ij}} s_{ijk} e_i e_i^T f(w_{ijk} x_j e_j + \tilde{b}_{ijk}) \\
        &=
        \sum_{i=1}^d \sum_{j=1}^d \sum_{k=1}^{N_{ij}} s_{ijk} E_{ii} f(w_{ijk} E_{jj} x + \tilde{b}_{ijk}). \nonumber
    \end{align}
    Here, $\tilde{b}_{ijk} = b_{ijk} e_j + \hat{b}_{ij}$. Since $s_{ijk} E_{ii}, w_{ijk} E_{jj} \in \mathbb{R}^{d\times d}$ are diagonal matrices and $\tilde{b}_{ijk} \in \mathbb{R}^d$, we have that $h \in \text{span}(\mathcal{F}_{\text{diag}}(f))$. Therefore, the proof is finished since $g \in {\text{cl(span}(\mathcal{F}_{\text{diag}}(f)))}$ according to the UAP of neural networks.

    (2) $\mathcal{H}(\mathcal{F}_{\text{rest}}(f))$ has the $C(\Omega)$-UAP for $\mathrm{Diff}_0(\mathbb{R}^d)$. Recall that the restricted affine invariant family $\mathcal{F}_{\text{rest}}(f)$ is the set of all functions $ x \mapsto \tilde D f(\tilde A x+b)$, where $b \in \mathbb{R}^d$ and $\tilde D,\tilde A \in \mathbb{R}^{d\times d}$ are diagonal matrices, and the restrictions are as follows. 
    1) The entries of $\tilde D$ are $\pm 1$ or 0, and 
    2) the entries of $\tilde A$ are smaller than or equal to 1. The restriction imposed on $\tilde D$ can be removed easily, whereas the restriction imposed on $\tilde A$ requires a scaling technique.

    Continuing from Eq.~\eqref{eq:h_in_diagonal_affine}, we denote $v(x) = s_{ijk} E_{ii} f(w_{ijk} E_{jj} x + \tilde{b}_{ijk})$, choose a number $\lambda > \max\{1,|w_{ijk}|\}$ and define
    $$
    \tilde{v}(x) = \text{sgn}(s_{ijk}) E_{ii} f(\tfrac{w_{ijk}}{\lambda} E_{jj} x + \tilde{b}_{ijk}).
    $$ 
    Then, we have that $\tilde v \in \mathcal{F}_{\text{rest}}(f)$. In addition, we define a linear function $\Lambda = \Lambda_{j,\lambda}$ whose $j$-th coordinate is scaled by $\lambda$:
    \begin{align*}
        \Lambda_{j,\lambda}: x=(x_1,...,x_d) \mapsto (x_1,...,x_{j-1}, \lambda x_j, x_{j+1},...,x_d).
    \end{align*}
    Then, $v(x) = |s_{ijk}| \tilde v(\Lambda(x)) = |s_{ijk}| \lambda_{ij} \Lambda^{-1} \circ \tilde v \circ \Lambda(x)$, where $\lambda_{ij}=\Lambda_{j,\lambda}(e_i)$; this equals $\lambda$ for $i=j$ and $1$ for $i\neq j$. Given that $\Lambda_{j,\lambda}$ is a flow map $\Lambda_{j,\lambda} = \phi^{\ln(\lambda)}_{E_{jj}\cdot}$, we can represent $\phi^\tau_v$ as the following composition:
    \begin{align}
        \phi^\tau_v =
        \phi^{ |s_{ijk}| \lambda_{ij} \tau}_{\Lambda^{-1} \circ \tilde v \circ \Lambda}
        &= \Lambda^{-1} \circ \phi^{ |s_{ijk}|  \lambda_{ij} \tau}_{\tilde v } \circ \Lambda \notag \\
        & = \phi^{\ln(\lambda)}_{-E_{jj}\cdot} \circ \phi^{ |s_{ijk}|  \lambda_{ij} \tau}_{\tilde v } \circ \phi^{\ln(\lambda)}_{E_{jj}\cdot}.
    \end{align}
    According to (1) and Proposition~\ref{th:composition_approximation}, we only need to prove that all flow maps $\phi^{\tau}_{\pm E_{jj}\cdot}$ with $\tau>0$ and $j=1,...,d$ belong to $\text{cl}{(\mathcal{H}(\mathcal{F}_{\text{rest}}(f)))}$. This can be achieved by showing that the linear map $E_{jj}\cdot : x \mapsto E_{jj} x = x_j e_j$ is in the closure of $\text{span}(\mathcal{F}^*_{j})$ under the topology of $C(\mathcal{K}, \mathbb{R}^d)$, where $\mathcal{F}^*_j \subset \mathcal{F}_{\text{rest}}(f)$ is given by
    \begin{align*}
        \mathcal{F}^*_{j} = \{
            x \mapsto \pm E_{jj} f(\mu E_{jj} x +b_j) ~|~
            \mu \in [-1,1], b_j\in \mathbb{R}^d
            \}.
    \end{align*}
    
    Next, we prove that $E_{jj}\cdot \in {\text{cl(span}(\mathcal{F}^*_{j}))}$, which finishes the proof. The idea is simple: zooming in on $f$ near some points can lead to a function that is close to linear. This idea is formally stated below.

    1) Consider a scalar function $p_j:\mathbb{R}\to\mathbb{R}$, which maps $x_j$ to $e_j^T f(x_j e_j + \tilde b_j)$. Here, $\tilde b_j \in \mathbb{R}^d$ is chosen such that $p_j$ is nonlinear according to the fully coordinate nonlinearity of $f$. Since Lipschitz continuous functions are differentiable almost everywhere, there is a point $\xi \in \mathbb{R}$ such that the derivative $p_j'(\xi)$ exists and is nonzero. Zooming in on $p_j$ near $\xi$ enables us to approximate the linear function with a slope of $p_j'(\xi)$. Specifically, for any $\varepsilon>0$ and $M>1$, there exists a number $\delta\in (0,1)$ such that
    \begin{align}
        \Big| \frac{ p_j(\tilde x_j) - p_j(\xi)}{\tilde x_j - \xi}- p_j'(\xi) \Big| 
        \le \frac{|p_j'(\xi)| \varepsilon}{M}, 
    \end{align}
    for all $\tilde x_j \in (\xi-\delta, \xi + \delta), \tilde x_j \neq \xi$.
    
    Let $\tilde x_j = \xi + \tfrac{\delta}{M} x_j$, $x_j \in [-M,M]$; then, we have that
    \begin{align}
        \label{eq:inequality_linearize_p}
        \Big| \tfrac{1}{p_j'(\xi)} \tfrac{M}{\delta} \big( p_j(\xi + \tfrac{\delta}{M} x_j) - p_j(\xi) \big) - x_j \Big| \le \tfrac{|x_j|}{M} \varepsilon \le  \varepsilon, \notag
        \\
        \forall  x_j \in [-M, M].
    \end{align}

    2) The fully coordinate nonlinearity of $f$ also implies that there is a vector $\hat b_j \in \mathbb{R}^d$ such that $e_{j}^T f(\hat b_j)=:a_j\neq 0$. The inequality in Eq.~\eqref{eq:inequality_linearize_p} then indicates that $E_{jj}\cdot : x \mapsto x_j e_j$ can be arbitrarily approximated by functions $g_j$ of the following form:
    \begin{align*}
        g_j : x \mapsto d_j E_{jj} f( \mu E_{jj} x + b_j ) + 
        c_j E_{jj} f(\hat b_j), \notag \\
        x \in [-M,M]^d \supset \mathcal{K},
    \end{align*}
    where $d_j = \tfrac{1}{p_j'(\xi)} \tfrac{M}{\delta}, \mu = \tfrac{\delta}{M} \in (0,1), b_j = \xi e_j + \tilde b_j$ and $c_j = - \tfrac{d_j p_j(\xi)}{a_j}$. Since $\mu \in (0,1)$, both $E_{jj} f( \mu E_{jj} \cdot + b_j )$ and $E_{jj} f(\hat b_j)$ = $E_{jj} f(0\cdot + \hat b_j)$ are functions in $\mathcal{F}^*_j$; thus, we have that $g_j \in \text{span}(\mathcal{F}^*_j)$ and further that $E_{jj}\cdot \in {\text{cl(span}(\mathcal{F}^*_j))}$. The proof can now be finished.
\end{proof}

\section{Discussion}
\label{sec:discussion}

\subsection{Lie bracket-generating families}
\label{sec:bracket-generating}

Notably, if the control family $\mathcal{F}$ is smooth (i.e., all functions contained in $\mathcal{F}$ are smooth), then the $\text{span}(\mathcal{F})$ concept in Theorem \ref{th:UAP_span_F_to_HF} can naturally be extended to $\text{Lie}(\mathcal{F})$. This extension follows directly from the well-known Chow--Rashevskii theorem on controllability \cite{chow1940systeme}. Specifically, we present the following definition.

\begin{definition}[Lie bracket-generating family]
        \label{def:Lie_family}
        We say that a smooth family $\mathcal{F}$ is a Lie bracket-generating family on $\mathbb{R}^d$ if $\text{Lie}(\mathcal{F}) |_x = \mathbb{R}^d$ for all $x \in \mathbb{R}^d$.
\end{definition}
Here, $\text{Lie}(\mathcal{F})$ is the span of all vector fields of $\mathcal{F}$ and their iterated Lie brackets (of any order):
\begin{align*}
    \text{Lie}(\mathcal{F})
    &=
    \text{span}\{
        f_1, [f_1,f_2],[f_1,[f_2,f_3]],... |
        f_1, f_2, f_3,... \in \mathcal{F}
        \} \notag
\end{align*}
where the operator $[f,g](x) = \nabla g(x) f(x) - \nabla f(x) g(x)$ is the Lie bracket between two smooth vector fields $f$ and $g$.
The flow map $\phi^\tau_{[f,g]}$ of $[f,g]$ can be approximated by
$
\phi^{\sqrt{\tau}}_{-g} \circ
\phi^{\sqrt{\tau}}_{-f} \circ
\phi^{\sqrt{\tau}}_{g} \circ
\phi^{\sqrt{\tau}}_{f}
$
when $\tau$ is small. As a consequence, we have the following Corollary~\ref{th:UAP_Lie_F_to_HF} according to Theorem \ref{th:UAP_span_F_to_HF}.

\begin{corollary}
    \label{th:UAP_Lie_F_to_HF}
    Let $\Omega \subset \mathbb{R}^d$ be a compact domain and let $\mathcal{F}$ be a smooth, symmetric and Lipschitz continuous control family. If $\text{Lie}(\mathcal{F})$ is dense in $C(\mathcal{K},\mathbb{R}^d)$ for every compact domain $\mathcal{K} \subset \mathbb{R}^d$, then $\mathcal{H}(\mathcal{F})$ has the $C(\Omega)$-UAP for $\text{Diff}_0(\mathbb{R}^d)$.
\end{corollary}

Note that $\text{Lie}(\mathcal{F})$ can be further generalized to Lie brackets of Lipschitz continuous vector fields; it is interesting to investigate the UAP of control systems driven by Lipschitz continuous vector fields as in \cite{rampazzo2001set}.

\subsection{The meaning of a ``minimal control family"}

Since the UAP of $\mathcal{H}(\mathcal{F}_1)$ for nonlinear functions no longer holds when the $\text{ReLU}$ function is removed from $\mathcal{F}_1$, we refer to $\mathcal{F}_1$ as a minimal control family that includes all affine maps for universal approximation purposes (when $d \ge 2$). In this section, we explore the concept of a minimal control family in a stricter sense.

The control families $\mathcal{F}_0,\mathcal{F}_1,\mathcal{F}_2$ can be slimmed to finite sets since we have that $\mathcal{F}_0 = \text{span}(\mathcal{F}^*_0)$:
\begin{align}
    \mathcal{F}^*_0 = 
    \{ x \mapsto E_{ij}x \text{ or } x \mapsto e_i ~|~ i,j=1,2,...,d\},
\end{align}
where $e_i \in \mathbb{R}^d$ is the $i$-th unit coordinate vector and $E_{ij}$ is a $d \times d$ matrix that has zeros for all its entries except for a 1 at index $(i, j)$. Letting $\mathcal{F}^*_1 :=  \mathcal{F}^*_0 \cup \{\text{ReLU}(\cdot)\}$ and $\mathcal{F}^*_2 :=  \mathcal{F}^*_0 \cup \{\pm \text{ReLU}(\cdot)\}$, we have the following Corollary~\ref{th:th:UAP_ReLU_family_corollary} according to Theorem~\ref{th:UAP_ReLU_family}.

\begin{corollary}
    \label{th:th:UAP_ReLU_family_corollary}
    $\mathcal{H}(\mathcal{F}^*_2)$ has the $C$-UAP for $\mathrm{Diff}_0(\mathbb{R}^d)$. In addition, $\mathcal{H}(\mathcal{F}^*_1)$ has the $C$-UAP for $\mathrm{Diff}_0(\mathbb{R}^d)$ if and only if $d\ge 2$.
\end{corollary}

Note that since both $\mathcal{F}^*_1$ and $\mathcal{F}^*_2$ contain finite numbers of functions, one can examine their subsets individually to identify strictly minimal families, where removing any function results in the loss of the UAP. Constructing such families is an intriguing task that we leave for interested readers to explore.

Furthermore, as stated in Proposition 4.4 of \cite{Cuchiero2020Deep}, there exists a family $\hat{\mathcal{F}}_0$ containing only four functions from which all linear vector fields can be generated based on $\text{Lie}(\hat{\mathcal{F}}_0)$. This suggests that the number of elements contained in a strictly minimal family satisfying the UAP can be very small (e.g., fewer than six). Investigating the minimal number of elements contained in such families presents an interesting direction for future research.

\subsection{Affine control system}

Consider the following affine control system associated with $\mathcal{F}^*_1 :=  \mathcal{F}^*_0 \cup \{\text{ReLU}(\cdot)\}$:
\begin{align}
    \label{eq:affine_control_system}
    \dot x(t) = \sum_{i=1}^{m} u_i(t) f_i(x), \quad
    f_i \in 
    \mathcal{F}^*_1,
\end{align}
where $m=d^2+d+1$ and $u=(u_1,...,u_m):[0,\infty)\to \mathbb{R}^m$ is the control; then, Theorem~\ref{th:UAP_ReLU_family} leads to the following Corollary~\ref{th:th:UAP_ReLU_family_corollary_2}.

\begin{corollary}
        \label{th:th:UAP_ReLU_family_corollary_2}
For any OP diffeomorphism $\Phi$ of $\mathbb{R}^d, d\ge 2$, a compact domain $\Omega \subset \mathbb{R}^d$ and $\varepsilon>0$, there exists a piecewise constant control $u(t)$ and a time $\tau$ such that the flow map $\phi^\tau_{u}$ of Eq.~\eqref{eq:affine_control_system} uniformly approximates $\Phi$ with an accuracy below $\varepsilon$; \emph{i.e.}, 
$\|\Phi(x)-\phi^\tau_{u}(x)\| \le \varepsilon$ for all $x \in \Omega$.
\end{corollary}

\subsection{The $L^p$-UAP for continuous functions}

As mentioned earlier, $\text{Diff}_0(\mathbb{R}^d)$ is dense in $C(\Omega,\mathbb{R}^d)$ under the $L^p$ norm provided that the dimensionality $d$ is larger than one \cite{Brenier2003Approximation}. Therefore, our theorems can be directly extended to the $L^p$-UAP for continuous functions in compact domains. For instance, the following Corollary~\ref{th:UAP_well_function} is implied by our Theorem~\ref{th:UAP_affine}.
\begin{corollary}
    \label{th:UAP_well_function}
    Let $d\ge2$, $F:\mathbb{R}^d\to \mathbb{R}^d$ be a continuous function, and let $\mathcal{F}$ be a restricted affine invariant family. If the closure of $\text{span}(\mathcal{F})$ contains a fully coordinate nonlinear Lipschitz continuous function $f: \mathbb{R}^d\to \mathbb{R}^d$, then for any $p\in[1,\infty)$, a compact domain $\Omega \subset \mathbb{R}^d$ and $\varepsilon>0$, there exists an $\hat F \in \mathcal{H}(\mathcal{F})$ such that $\|\hat F-F\|_{L^p(\Omega)} \le \varepsilon.$
\end{corollary}

Note that the main theorem in \cite{Li2022Deep} shows a similar result, but the function $f$ is required to be a so-called well function; \emph{i.e.}, the set $\{x\in \mathbb{R}^d ~|~ f_i(x)=0\}$ is bounded and convex for each component $f_i$ of $f$. It is easy to check that a well function $f$ is fully coordinate nonlinear. In fact, when a vector $b \in \mathcal{K}$ is chosen, $f_i(\cdot e_i + b) \in C(\mathbb{R},\mathbb{R})$ is nonlinear for all $i=1,...,d$. The reason for this result is that if $f_i(\cdot e_i + b)$ is linear, then it must be zero, which contradicts the definition since it must be nonzero outside a compact domain.

\subsection{The $C$-UAP for continuous functions}

Utilizing topological theory, we can extend the $C$-UAP result to all continuous functions (rather than just diffeomorphisms) by leveraging additional dimensions.

In fact, Whitney's theorem from the field of differential topology (see, for example, \cite[P.26]{Hirsch1976Differential}) suggests that for any compact domain $\mathcal{K} \subset \mathbb{R}^{n_x}$, $f \in C(\mathcal{K},\mathbb{R}^{n_y})$, $d \ge \max(2n_x + 1, n_y)$, and any $\varepsilon > 0$, there exists an OP diffeomorphism $\Phi$ on $\mathbb{R}^d$, along with two linear maps $\alpha: \mathbb{R}^{n_x} \to \mathbb{R}^d$ and $\beta: \mathbb{R}^d \to \mathbb{R}^{n_y}$ such that
\begin{align}
\|f - \beta \circ \Phi \circ \alpha\|_{C(\mathcal{K})} \le \varepsilon .
\end{align}

Building on this fact, our results can be directly extended to the $C$-UAP of the flowmap space ($\mathcal{H}(\mathcal{F}_1)$, for example) for continuous functions by utilizing additional dimensions. For instance, the following corollary (using the above notation) follows from Theorem~\ref{th:UAP_ReLU_family}.

\begin{corollary}
Let $d \ge \max(2 n_x+1,n_y)$; then, the hypothesis space
$\mathcal{H} = \{\beta \circ \Phi \circ \alpha | \Phi \in \mathcal{H}(\mathcal{F}_1)\}$
has the $C(\mathcal{K})$-UAP for $C(\mathcal{K},\mathbb{R}^{n_y})$.
\end{corollary}

\subsection{Universal interpolation property}

Another interesting topic related to the UAP of dynamical systems is the universal interpolation property (UIP) \cite{Cuchiero2020Deep}. Unlike the UAP, the UIP requires the examined dynamical system to interpolate any continuous function on a compact domain, and the two properties cannot be derived from each other \cite{Cheng2023Interpolation}. Exploring the UIP possessed by dynamical systems associated with more general control families would be an intriguing direction for future research.

\section{Conclusion}
\label{sec:summary}

In this paper, we introduce several sufficient conditions for the control family $\mathcal{F}$ such that the flow maps composing the dynamical systems associated with $\mathcal{F}$ can uniformly approximate all OP diffeomorphisms of $\mathbb{R}^d$ on any compact domain $\Omega$. In particular, the UAP is verified for the minimal control families $\mathcal{F}_1$ and $\mathcal{F}_2$.

In terms of OP diffeomorphisms, our results can be directly extended to the $L^p$-UAP and $C$-UAP for continuous maps.

However, the constructed control families, i.e., $\mathcal{F}_1$, $\mathcal{F}_1^*$, and $\mathcal{F}_2$, are very special. Although the employed ReLU function can be directly extended to a wide range of functions, such as softplus, leaky-ReLU, or other ReLU-like functions, more general finite families $\mathcal{F}^*=\{f_1,...,f_m\}$ will be considered so that the flows of an affine control system $\dot x(t) = \sum_{i=1}^{m} u_i(t) f_i(x)$ with a control $u = (u_1,...,u_m)$ are universal approximators for continuous functions in compact domains. The Chow--Rashevskii theorem in control theory suggests considering $\mathcal{F}^*$ as a Lie bracket-generating family on $\mathbb{R}^d$ \cite{Boscain2019Introduction}. It is interesting to examine whether assuming that $\mathcal{F}^*$ further contains a nonlinear function is sufficient. We leave this question for future work.

\appendix

\subsection{Proof of Lemma~ \ref{th:ODE_error_estimation}}
\begin{proof}
    First, let us consider the case in which $\|f_1(x)-f_2(x)\|< \delta$ for all $x \in \mathbb{R}^d$. We define the error $e(x,t)=\phi_{f_1}^t(x) - \phi_{f_2}^t(x),$ and use the integral form of an ODE:
    \begin{align}
        \phi_{f_i}^t(x) = \phi_i^0(x) + \int_0^t f_i(\phi_{f_i}^s(x)) d s,
    \end{align}
    Then, we have that
    \begin{align}
        \|e(x,t)\|
        = 
            & \Big\|\int_0^t 
            f_1(\phi_{f_1}^s(x)) - f_2(\phi_{f_2}^s(x)) ds
            \Big\| \nonumber\\
        \le 
            &\int_0^t \Big(
            \|f_1(\phi_{f_1}^s(x))-f_1\phi_{f_2}^s(x))\| \nonumber\\
            & +
            \|f_1(\phi_{f_2}^s(x))-f_2(\phi_{f_2}^s(x))\|
            \Big)ds \nonumber\\
        \le 
            &\int_0^t 
            (L \|e(x,s)\|+\delta)
            ds \nonumber\\
        \le 
            &\delta t + L \int_0^t \|e(x,s)\| d s. 
    \end{align}
    By employing the well-known Gronwall inequality, we have that $\|e(x,t)\| \le \delta t e^{L t}$. Then, letting $\delta = \min(1,\frac{\varepsilon}{\tau e^{L \tau}})$, we have that $\|e(x,t)\| \le \varepsilon$ for all $x \in \mathbb{R}^d$ and $t \in [0,\tau]$.
    
    Now, we relax $x \in \mathbb{R}^d$ to $x \in \Omega_{\tau}$. This relaxation only requires that $\phi_{f_i}^s(x) \in \Omega_\tau$
    for all $x \in \Omega$ and $s \in [0,\tau]$. Since $f_1(x)$ is Lipschitz continuous and bounded on $\Omega$, we have the following:
    \begin{align}
        \|\phi_{f_1}^t(x) -x \| 
        \le 
        &\int_0^t \|f_1(\phi_{f_1}^s(x))-f_1(x)\|ds 
        +\int_0^t\|f_1(x)\|ds \nonumber\\
        \le 
        &L\int_0^t\|\phi_{f_1}^s(x)-x\|ds + V t \nonumber\\
        \le 
        & V t e^{L t} 
        \le 
        V \tau e^{L \tau}.
\end{align}
    In addition to $\|e(x,t)\| \le \delta t e^{L t} \le \tau e^{L \tau}$, we have that $\|\phi_{f_2}^t(x) -x\| \le (V+1) \tau e^{L \tau}$.
    Therefore, both $\phi_{f_1}^t(x)$ and $\phi_{f_2}^t(x)$ are in the domain $\Omega_\tau$. This completes the proof.
\end{proof}

\subsection{Proof of Lemma~\ref{th:split_method}}
\begin{proof}
    We only prove the case in which $m=2$, while the general cases concerning $m$ can be proved in the same way. In addition, since $v$ is piecewise constant with respect to $t$, $x(\tau)$ is a composition of finite flow maps
$x(\tau) = \phi^\tau_{v(\cdot,t_q)} \circ \cdots \circ \phi^\tau_{v(\cdot,t_1)}$,
where $t_1,...,t_q$ are the break points. According to Proposition~\ref{th:composition_approximation}, we only need to approximate each $\phi^\tau_{v(\cdot,t_q)}$. For this reason, it is sufficient to provide a proof for the case in which $v_i$ is constant for $t$ within the time interval $[0,\tau]$, \emph{i.e.,} $v_i(x,t) = f_i(x)$ and $v(x,t) = f(x)=f_1(x)+f_2(x)$. 
In this case, the iterative process related to $x_k$ is
    \begin{align}
        x_{k+1} &= x_k + \Delta t f_1(x_k) + \Delta t f_2(x_k + \Delta t f_1(x_k)).
    \end{align}
    We define $ \hat x(t) = x_k $ for $t \in [t_k, t_{k+1})$; then, we have that
    \begin{align}
        \hat x(t)
            &= x_k \notag \\  
            &= x_{k-1} + \int_{t_{k-1}}^{t_k} \big( f_1( \hat x(s)) + f_2( \hat x(s) \notag + \Delta t f_1( \hat x(s))) \big) d s\\
            &= x_0 + \int_0^{t_k} \big( f_1( \hat x(s)) + f_2( \hat x(s) + \Delta t f_1( \hat x(s))) \big) d s, \nonumber\\
        x(t) &=
            x_0 + \int_0^t \big( f_1(x(s)) + f_2(x(s) \big) d s.
    \end{align}
    Then, the error $e(t) = x(t) - \hat x(t)$ can be estimated by using the Lipschitz constant:
    \begin{align}
        \|e(t)\| \le  
        & \int_{t_k}^t \| f(x(s)) \| d s \notag \\ 
        &+ \int_0^{t_k} \big( 2 L \|e(s)\| + L \Delta t \|f_1(\hat x(s))\| \big) ds \nonumber\\ 
        \le
        & 2 M \Delta t + \int_0^t \big( (2 L+ L^2 \Delta t) \|e(s)\| \notag \\ 
        & + L \Delta t \|f_1(x(s))\| \big) ds \nonumber\\
        \le 
        &(2 + L t ) M \Delta t + \int_0^t  (2 + L \Delta t) L \|e(s)\| ds,
    \end{align}
    where $M = \sup \{ \|f_i(\phi^s_f(x_0))\| ~|~ s \in [0,\tau], x_0 \in \Omega, i=1,...,d\}$ is a well-defined number. Let $\Delta t$ be sufficiently small such that $L \Delta t \le 1/2$; then, the estimation process simply becomes
    \begin{align}
        \|e(t)\| 
        & \le m (1 + L t ) M \Delta t + \int_0^t (2m-1) L \|e(s)\| ds \notag\\
        & \le  m  M  \text{e}^{ 2 m L t} \Delta t.
    \end{align}
    Here, the Gronwall inequality is employed again, and the estimation procedure is valid for all cases in which $m\ge 2$. Next, let $\Delta t \le \min \big\{ \tfrac{1}{2L}, \tfrac{\text{e}^{-2 m L \tau} \varepsilon}{m M} \big\}$; then, the inequality $\|x(\tau) - x_n\| \le \varepsilon$ is satisfied for all $x_0 \in \Omega$.
\end{proof}

\section*{Acknowledgments}

We would like to thank Prof. Paulo Tabuada and Prof. Bahman Gharesifard for insightful discussions on their results.

We also thank the anonymous reviewers for their valuable comments and useful suggestions.
This work was funded by the National Natural Science Foundation of China under grants 12494543 and 12171043 and the Fundamental Research Funds for the Central Universities.

\bibliography{refs}
\bibliographystyle{plain}



\end{document}